\documentclass[a4paper]{article}

\usepackage[english]{babel}
\usepackage[utf8x]{inputenc}
\usepackage[T1]{fontenc}
\usepackage{chicago}

\usepackage[a4paper,top=3cm,bottom=2cm,left=3cm,right=3cm,marginparwidth=1.75cm]{geometry}

\usepackage{microtype}
\usepackage{graphicx}
\usepackage{subfigure}
\usepackage{booktabs} 

\usepackage{amsmath}
\usepackage{amssymb}
\usepackage{graphicx}
\usepackage{caption}
\usepackage{amsthm}
\newtheorem{theorem}{Theorem}
\newtheorem{lemma}{Lemma}

\newtheorem*{corollary}{Corollary}
\newcommand\numberthis{\addtocounter{equation}{1}\tag{\theequation}}

\DeclareMathOperator{\pfaffian}{Pf}
\usepackage[titletoc,title]{appendix}
\allowdisplaybreaks
\usepackage[hyphens]{url}
\usepackage{enumitem,kantlipsum}

\usepackage{xcite}

\usepackage{xr,refcount}
\makeatletter
\newcommand*{\addFileDependency}[1]{
  \typeout{(#1)}
  \@addtofilelist{#1}
  \IfFileExists{#1}{}{\typeout{No file #1.}}
}
\makeatother

\newcommand*{\myexternaldocument}[1]{%
    \externaldocument{#1}%
    \addFileDependency{#1.tex}%
    \addFileDependency{#1.aux}%
}

\myexternaldocument{icml}
\addtocounter{theorem}{0}
\addtocounter{lemma}{0}

\usepackage{algorithm}
\usepackage[noend]{algpseudocode}

\title{\bf Inference and Sampling of $K_{33}$-free Ising Models}
\date{}
\author{Valerii Likhosherstov$^{1}$, Yury Maximov$^{(1,2)}$ and Michael Chertkov$^{(1,2,3)}$\\
$^{1}$ Skolkovo Institute of Science and Technology, Moscow, Russia\\
$^{2}$ Theoretical Division and Center for Nonlinear Studies,\\ Los Alamos National Laboratory, Los Alamos, NM, USA\\
$^{3}$ Graduate Program in Applied Mathematics,\\ University of Arizona, Tucson, AZ, USA}

\usepackage[colorlinks=true
]{hyperref}

\begin{document}

\maketitle

\begin{abstract}
We call an Ising model tractable when it is possible to compute its partition function value (statistical inference) in polynomial time. The tractability also implies an ability to sample configurations of this model in polynomial time. The notion of tractability extends the basic case of planar zero-field Ising models. Our starting point is to describe algorithms for the basic case, computing partition function and sampling efficiently.
Then, we extend our tractable inference and sampling algorithms to models whose triconnected components are either planar or graphs of $O(1)$ size. In particular, it results in a polynomial-time inference and sampling algorithms for $K_{33}$ (minor)-free topologies of zero-field Ising models---a generalization of planar graphs with a potentially unbounded genus. \footnote{The paper to appear at the \emph{Proceedings of the 36-th International Conference on Machine Learning}, Long Beach, California, PMLR 97, 2019.  Implementation of the algorithms is available at \href{https://github.com/ValeryTyumen/planar_ising}{https://github.com/ValeryTyumen/planar\_ising}.}
\end{abstract}

\section{Introduction}

Computing the partition function of the Ising model is generally intractable, even an approximate solution in the special anti-ferromagnetic case of arbitrary topology would have colossal consequences in the complexity theory \cite{jerrum-sinclair}. Therefore, a question of interest---rather than addressing the general case---is to look after tractable families of Ising models. In the following, we briefly review tractability related to planar graphs and graphs embedded in surfaces of small genus.

\textbf{Related work.} Onsager \cite{onsager} gave a closed-form solution for the partition function in the case of a homogeneous interaction Ising model over an infinite two-dimensional square grid without a magnetic field.
This result has opened an exciting era of phase transition discoveries, which is arguably one of the most significant contributions in theoretical and mathematical physics of the 20th century.
Then, Kac and Ward \cite{kac-ward} showed in the case of a finite square lattice that the problem of the partition function computation is reducible to a determinant. Kasteleyn \cite{kasteleyn} generalized the results to the case of an arbitrary inhomogeneous interaction Ising model over an arbitrary planar graph. Kasteleyn's construction was based on mapping of the Ising model to a perfect matching (PM) model with specially defined weights over a modified graph. Kasteleyn's construction was also based on the so-called Pfaffian orientation, which allows counting of PMs by finding a single Pfaffian (or determinant) of a matrix. Fisher \cite{fisher} simplified Kasteleyn's construction such that the modified graph remained planar. Transition to PM is fruitful because it extends planar zero-field Ising model inference to models embedded on a torus \cite{kasteleyn} and, in fact, on any surface of small (orientable) genus $g$, but with a price of the additional, multiplicative, and exponential in genus, $4^g$, factor in the algorithm's run time~\cite{gallucio}.

A parallel way of reducing the planar zero-field Ising model to a PM problem consists of constructing a so-called expanded dual graph \cite{bieche,barahona,schraudolph-kamenetsky}. 
This approach is more natural and interpretable because there is a one-to-one correspondence between spin configurations and PMs on the expanded dual graph. An extra advantage of this approach is that the reduction allows one to develop an exact efficient sampling. Based on linear algebra and planar separator theory \cite{lipton-tarjan}, Wilson introduced an algorithm \cite{wilson} that allows one to sample PMs over planar graphs in $O(N^\frac32)$ time. The algorithms were implemented in \cite{thomas-middleton1,thomas-middleton2} for the Ising model sampling, however, the implementation was limited to only the special case of a square lattice. In \cite{thomas-middleton1} a simple extension of the Wilson's algorithm to the case of bounded genus graphs was also suggested, again with the $4^g$ factor in complexity.
Notice that imposing zero field condition is critical, as otherwise, the Ising model over a planar graph is  NP-hard \cite{barahona}. On the other hand, even in the case of zero magnetic field Ising models over general graphs are difficult \cite{barahona}.

\textbf{Contribution.} In this manuscript, we discuss tractability related to the Ising model with zero magnetic fields over graphs more general than planar. Our construction is related to graphs characterized in terms of their excluded minor property.
Planar graphs are characterized by excluded $K_5$ minor and $K_{33}$ minor (Wagner's theorem \cite{diestel}, Chapter 4.4). Therefore, instead of attempting to generalize from planar to graphs embedded into surfaces of higher genus, it is natural to consider generalizations associated with a family of graphs excluding $K_5$ minor or $K_{33}$ minor. 

In this manuscript, we show that $K_{33}$-free zero-field Ising models are tractable in terms of inference and sampling and give a tight asymptotic bound, $O(N^\frac32)$, for both operations. For that purpose, we use graph decomposition into triconnected components---the result of recursive splitting by pairs of vertices, disconnecting the graph.
Indeed, the $K_{33}$-free graphs are simple to work with because their triconnected components are either planar or $K_5$ graphs \cite{hall}.
Therefore, the essence of our construction is to decompose the inference task in Ising over a $K_{33}$-free graph into a sequential dynamic programming evaluation over planar or $K_5$ graphs in the spirit of \cite{straub}. Notice that the triconnected classification of the tractable zero-field Ising models is complementary to the aforementioned small genus classification.  We illustrate the difference between the two classifications with an explicit example of a tractable problem over a graph with genus growing linearly with graph size. 

\textbf{Structure.} The manuscript is organized as follows. Sections \ref{sec:def} and \ref{sec:prob}, respectively, establish notations and pose problems of inference and sampling. Section \ref{sec:planar} presents transition from the zero-field Ising model to an equivalent tractable perfect matching (PM) model. This provides a description of a $O(N^\frac32)$ inference and sampling method in planar models, which is new (to the best of our knowledge), and it sets the stage for what follows.
Section \ref{sec:dynprog} discusses a scheme for polynomial inference and sampling in zero-field models over graphs with triconnected components that are either planar or of $O(1)$ size. Section \ref{sec:k33} applies this scheme to $K_{33}$-free zero-field Ising models, resulting in tight asymptotic bounds, which appear to be equivalent to those in the planar case.
Section \ref{sec:imp} describes benchmarks justifying correctness and efficiency of our algorithm.
Technical proofs of statements given throughout the manuscript can be found in the supplementary material. 

\section{Definitions and Notations} \label{sec:def}

Let $V = \{ v_1, ..., v_{| V |} \}$ be a finite set of \textit{vertices}, a multiset $E$ consisting of $e \subseteq V$, $| e | = 2$ be \textit{edges}, then we call $G = (V, E)$ a \textit{graph}. We call $G$ \textit{normal}, if $E$ is a set (i.e., there are no multiple edges in $G$).


A \textit{tree} is a connected graph without cycles.
For $V' \subseteq V$, let $G(V')$ denote a graph $(V', \{ \{ v, w \} \in E \, | \, v \in V', w \in V' \})$.
Let $H = (V_H, E_H)$ be a graph. Then $H$ is a \textit{subgraph} of $G$, if $V_H \subseteq V, E_H \subseteq E$.
Vertex $v \in V$ is an \textit{articulation point} of $G$, if $G(V \setminus \{ v \})$ is disconnected. $G$ is \textit{biconnected} if there are no articulation points in $G$. \textit{Biconnected component} is a maximal subgraph of $G$ without an articulation point.

The graph $G$ is \textit{planar} if it can be drawn on a plane without edge intersections. The corresponding drawing is referred to as \textit{planar embedding} of $G$. When no ambiguity arises, we do not distinguish planar graph $G$ from its embedding.

A set $E' \subseteq E$ is called a \textit{perfect matching} (PM) of $G$, if edges of $E'$ are disjoint and their union equals~$V$. 
$\text{PM}(G)$ denotes the set of all PMs of $G$. 
$K_p$ denotes a complete (normal) graph on $p$ vertices, and $K_{33}$ denotes a utility graph. \textit{Triple bond} is a graph of two vertices and three edges between them. \textit{Multiple bond} is a graph of two vertices and at least three edges between them.


\section{Problem Setup} \label{sec:prob}

Let $G = (V, E)$ be a normal graph, $| V | = N$. For each $v \in V$, define a random binary variable (a spin) $s_v \in \{ -1, +1 \}$, $S = (s_{v_1}, ..., s_{v_N})$.  Subscript $i$ will be used as shorthand for $v_i$, for brevity, thus  $S=(s_1, ..., s_N)$. For each $e \in E$, define a \textit{pairwise interaction} $J_e \in \mathbb{R}$. We associate assignment $X = (x_1, ..., x_N) \in \{ -1, +1 \}^N$ to vector $S$ with probability as follows:
\begin{equation}
\mathbb{P} (S = X) = \frac{1}{Z} \exp \left( \sum_{e = \{v, w\} \in E} J_e x_v x_w \right),
\label{eq:distr}
\end{equation}
where
\begin{equation*}
Z = \sum_{X \in \{ -1, +1 \}^N} \exp \left( \sum_{e = \{v, w\} \in E} J_e x_v x_w \right).
\end{equation*}

The probability distribution (\ref{eq:distr}) defines the so-called \textit{zero-field (or pairwise) Ising model}, and $Z$ is called the \textit{partition function} (PF) of the zero-field Ising (ZFI) model. Notice that $\mathbb{P}(S = X) = \mathbb{P}(S = -X)$.

Given a ZFI model, our goal is to find $Z$ (\textit{inference}) and draw samples from the model efficiently.

\section{Reducing Planar ZFI Model to PM Model} \label{sec:planar}

In this section, we consider a special case of planar graph $G$ and introduce a transition from the ZFI model to the perfect matching (PM) model on a different planar graph.

We assume that the planar embedding of $G$ is given
(and if not, it can be found in $O(N)$ time \cite{boyer}). We follow \cite{schraudolph-kamenetsky} in constructions discussed in this section.

\subsection{Expanded Dual Graph} \label{subsec:edg}

First, triangulate $G$ by adding new edges $e$ to $E$ such that $J_e = 0$. (The triangulation does not change probabilities of the spin assignments.) Graph $G$ is generated (use the same notation as for the original graph for convenience) and is biconnected with every face, including lying on the boundary, forming a triangle. Complexity of the triangulation procedure is $O(N)$, see \cite{schraudolph-kamenetsky} for an example. 

Second, construct a new graph, $G_F = (V_F, E_F)$, where each vertex $f$ of $V_F$ is a face of $G$, and there is an edge $e = \{ f_1, f_2 \}$ in $E_F$ if and only if $f_1$ and $f_2$ share an edge in $G$. By construction, $G_F$ is planar, and it is embedded in the same plane as $G$, so that each new edge $e = \{ f_1, f_2 \} \in E_F$ intersects the respective old edge. 
Call $G_F$ a \textit{dual graph} of $G$. Since $G$ is triangulated, each $f \in V_F$ has degree 3 in $G_F$. 

Third, obtain a planar graph $G^* = (V^*, E^*)$ and its embedding from $G_F$ by substituting each $f \in V_F$ by a $K_3$ triangle so that each vertex of the triangle is incident to one edge, going outside the triangle (see Figure \ref{fig:expanded_dual} for an illustration). Call $G^*$ an \textit{expanded dual graph} of $G$.

Newly introduced triangles of $G^*$, substituting $G_F$'s vertices, are called \textit{Fisher cities} \cite{fisher}. We refer to edges outside triangles as \textit{intercity edges} and denote their set as $E^*_I$. The set $E^* \setminus E^*_I$ of Fisher city edges is denoted as $E^*_C$. Notice that $e^* \in E^*_I$ intersects exactly one $e \in E$ and vice versa, which defines a bijection between $E^*_I$ and $E$; denote it by $g: E^*_I \to E$. Observe also that $| E^*_I | = | E | \leq 3 N - 6$, where $N$ is the size of $G$. Moreover, $E^*_I$ is a PM of $(V^*,E^*)$, and thus $| V^* | = 2 | E^*_I | = O(N)$. Since $G^*$ is planar, one also finds that $| E^* | = O(N)$. Constructing $G^*$ takes efforts of $O(N)$ complexity.

\begin{figure}[!h]
\centering
\subfigure[]{
  \centering
  \includegraphics[width=0.45\linewidth]{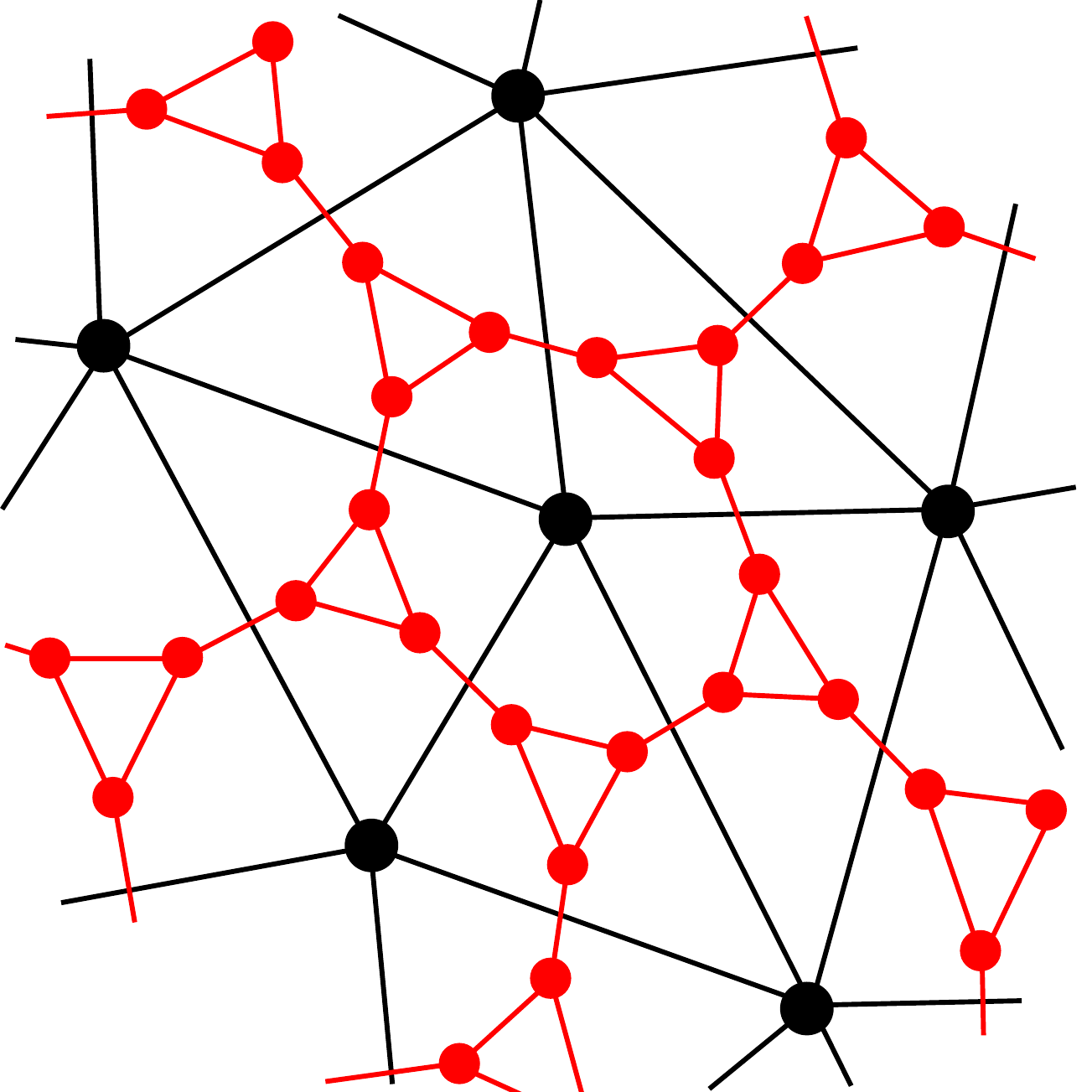}
  \label{fig:expanded_dual}
}%
\subfigure[]{
  \centering
  \includegraphics[width=0.45\linewidth]{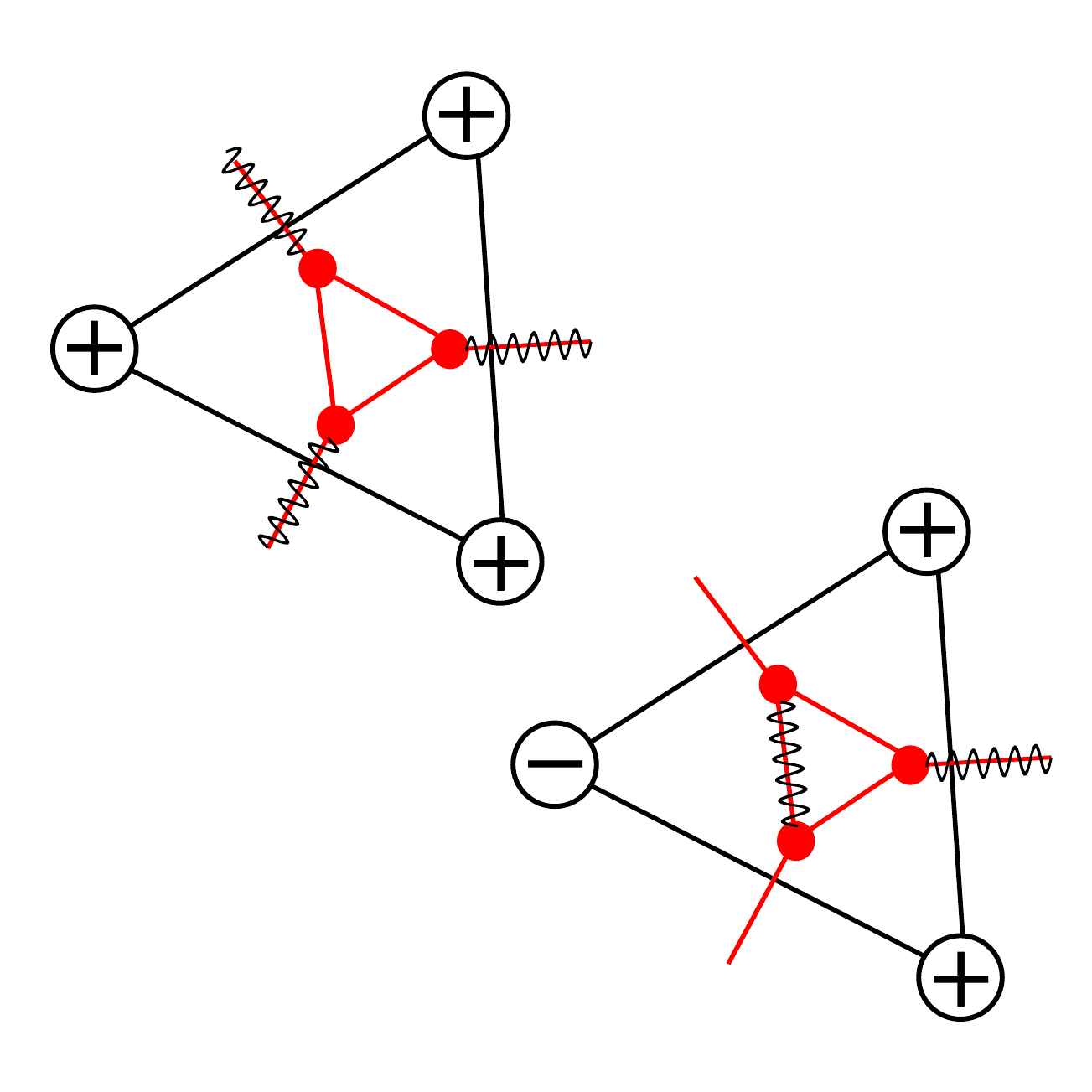}
  \label{fig:pm_cases}
}%
\caption{
(a) A fragment of $G$'s embedding after triangulation (black), expanded dual graph $G^*$ (red). (b) Possible $X$ configurations and corresponding $M(X)$ (wavy lines) on a single face of $G$. Rotation symmetric and reverse sign configurations are omitted.}
\end{figure}

\subsection{Perfect Matching (PM) Model} \label{subsec:pmc}

For $X \in \{ -1, +1 \}^N$, let $I(X)$ be a set $\{ e \in E^*_I \, | \, g(e) = \{ v, w \}, x_v = x_w \}$. Each Fisher city is incident to an odd number of edges in $I(X)$. Thus, $I(X)$ can be uniquely completed to a PM by edges from $E^*_C$. Denote the resulting PM by $M(X) \in \text{PM}(G^*)$ (see Figure \ref{fig:pm_cases} for an illustration). Let $\mathcal{C}_+ = \{ +1 \} \times \{ -1, +1 \}^{N - 1}$. 
\begin{lemma} \label{lemma:bij}
$M$ is a bijection between $\mathcal{C}_+$ and $\text{PM}(G^*)$.
\end{lemma}


Define weights on $G^*$ according to
\begin{equation*}
\forall e^* \in E^*: c_{e^*} = \begin{cases} \exp (2 J_{g(e^*)}), & e^* \in E^*_I \\ 1, & e^* \in E^*_C \end{cases}
\end{equation*}

\begin{lemma} \label{lemma:zfitopm}
For $E' \in \text{PM}(G^*)$ holds
\begin{equation}
    \mathbb{P} ( M(S) = E' ) = \frac{1}{Z^*} \prod_{e^* \in E'} c_{e^*}, \label{eq:pmprobs}
\end{equation}
where
\begin{equation}
    Z^* = \sum_{E' \in \text{PM}(G^*)} \prod_{e^* \in E'} c_{e^*} = \frac12 Z \exp\left( \sum_{e \in E} J_e\right)
    \label{eq:zstar}
\end{equation}
is the PF of the PM distribution (PM model) defined by (\ref{eq:pmprobs}).
\end{lemma}

Second transition of (\ref{eq:zstar}) reduces the $Z$ computation to solve for $Z^*$. Furthermore, only two equiprobable spin configurations $X'$ and $-X'$ (one of which is in $\mathcal{C}_+$) correspond to $E'$, and they can be recovered from $E'$ in $O(N)$ steps, thus resulting in the statement that one samples from (\ref{eq:distr}) if sampling from (\ref{eq:pmprobs}) is known.


The PM model can be defined for an arbitrary graph $\hat{G} = (\hat{V}, \hat{E}), \hat{N} = | \hat{V} |$ with positive weights $c_e, e \in E'$, as a probability distribution over $\hat{M} \in \text{PM}(\hat{G})$: $\mathbb{P}(\hat{M}) \propto \prod_{e \in \hat{M}} c_e$.



Our subsequent derivations are based on the following:
\begin{theorem} \label{th:pmmodel}
    Given the PM model defined on planar graph $\hat{G}$ of size $\hat{N}$ with positive edge weights $\{ c_e \}$, one can find its partition function and sample from it in $O(\hat{N}^\frac32)$ time.
\end{theorem}

Algorithms, constructively proving the theorem, are directly inferred from \cite{wilson,thomas-middleton1}, with minor changes/generalizations. Hence, we outline them in the supplementary material.

\begin{corollary}
     Inference and sampling of the PM model on $G^*$ (and, hence, the ZFI model on $G$) take $O(N^\frac32)$ time.
\end{corollary}

\section{Dynamic Programming within Triconnected Components} \label{sec:dynprog}

Starting with this section, we present new results. We describe a general algorithm that allows us to perform inference and sampling from the ZFI model in the case where the triconnected components of the underlying graph are either planar or of $O(1)$ size.

\subsection{Decomposition into Biconnected Components} \label{subsec:dec}

Consider a ZFI model (\ref{eq:distr}) over a normal graph $G = (V, E)$, $| V | = N$. If $G$ is disconnected, then distribution (\ref{eq:distr}) is decomposed into a product of terms associated with independent ZFI models over the connected components of $G$. Hence, we assume below, without loss of generality, that $G$ is connected.

Let $G_1, ..., G_h$ be biconnected components of $G$. They form a tree if an edge is drawn between $G_i$ and $G_j$ whenever $G_i$ and $G_j$ share an articulation point. A simple reduction (see supplementary material) shows that inference and sampling on $G$ are reduced to a series of inference and sampling on ZFI models induced by subgraphs $G_1, ..., G_h$.
\begin{lemma} \label{lemma:bic}
    Let $Z_1, ..., Z_h$ be partition functions of ZFI models induced by $G_1, ..., G_h$. Then,
    \begin{equation}
        Z = 2^{-h} Z_1 Z_2 ... Z_h
        \label{eq:bicinf}.
    \end{equation}
    Sampling from $\mathbb{P}(S = X)$ is reduced to a series of sampling on $G_1, ..., G_h$ and $O(N)$ post-processing.
\end{lemma}

Observe also that all the articulation points and the biconnected components of $G$ can be found in $O(N + | E |)$ steps \cite{hopcroft1}. Therefore later on, we assume without loss of generality that $G$ is biconnected.

\subsection{Biconnected Graph as a Tree of Triconnected Components} \label{subsec:tree}

In this subsection we follow \cite{hopcroft2,gutwenger}, see also \cite{08Mader} to define the tree of triconnected components. Following discussions of the previous subsection, one considers here a biconnected $G$. 

Let $v, w \in G$. Divide $E$ into equivalence classes $E_1, ..., E_k$ so that $e_1, e_2$ are in the same class if they lie on a common simple path that has $v, w$ as endpoints. $E_1, ..., E_k$ are referred to as \textit{separation classes}. If $k \geq 2$, then $\{ v, w \}$ is a \textit{separation pair} of $G$, unless (a) $k = 2$ and one of the classes is a single edge or (b) $k = 3$ and each class is a single edge. Graph $G$ is called \textit{triconnected} if it has no separation pairs.

\begin{figure*}[!t] \centering
\centering
\includegraphics[width=0.9\linewidth]{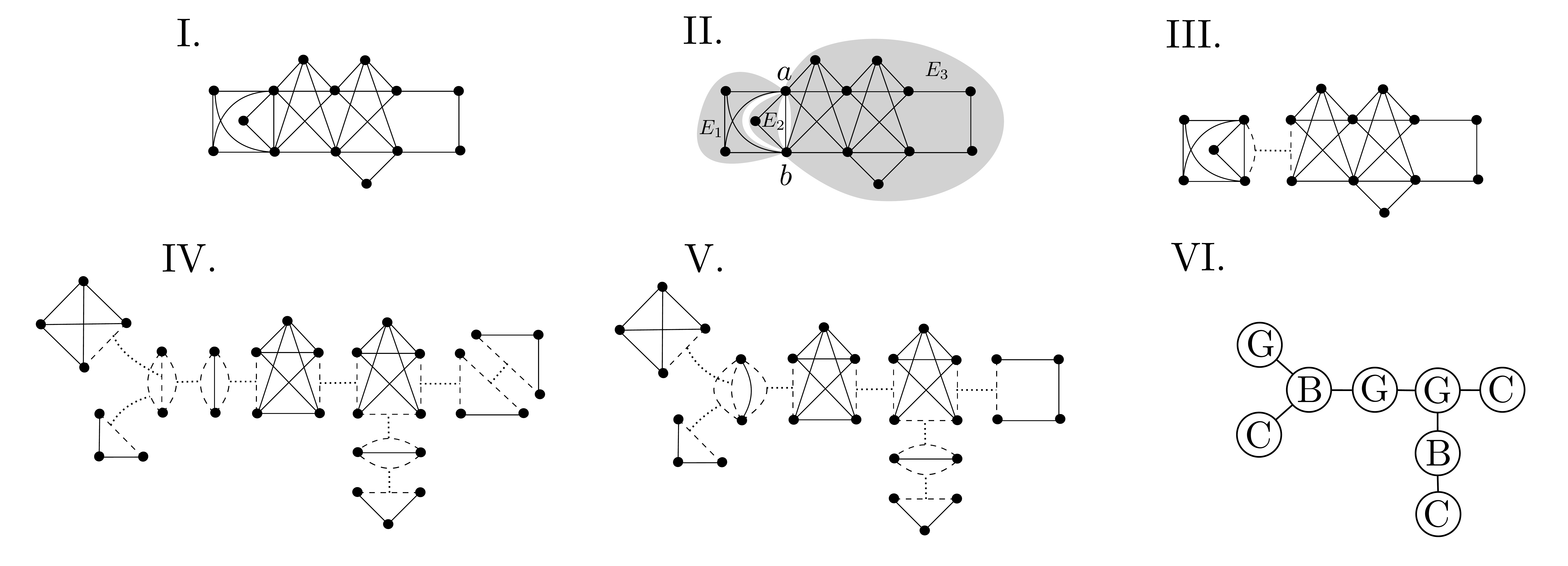}
\caption{(I) An example biconnected graph $G$. (II) A separation pair $\{ a, b \}$ of $G$ and separation classes $E_1, E_2, E_3$ associated with $\{ a, b \}$. (III) Result of split operation with $E' = E_1 \cup E_2, E'' = E_3$. Hereafter, dashed lines indicate virtual edges and dotted lines connect equivalent virtual edges in split graphs. (IV) Split components of $G$ (non-unique). (V) Triconnected components of $G$. (VI) Triconnected component tree $T$ of $G$; spacial alignment of V is preserved. ``G," ``B," and ``C" are examples of the  ``triconnected graph,"  ``multiple bond," and  ``cycle," respectively.}
\label{fig:gr_tree}
\end{figure*}

Let $\{ v, w \}$ be a separation pair in $G$ with equivalence classes $E_1, ..., E_k$. Let $E' = \cup_{i = 1}^l E_l, E'' = \cup_{i = l + 1}^k E_l$ be such that $| E' | \geq 2$, $| E'' | \geq 2$. Then, graphs $G_1 = (\cup_{e \in E'} e, E' \cup \{ e_\mathcal{V} \}), G_2 = (\cup_{e \in E''} e, E'' \cup \{ e_\mathcal{V} \})$ are called \textit{split graphs} of $G$ with respect to $\{ v, w \}$, and $e_\mathcal{V}$ is a \textit{virtual edge}, which is a  new edge between $v$ and $w$, identifying the split operation. Due to the addition of $e_\mathcal{V}$, $G_1$ and $G_2$ are not normal in general.

Split $G$ into $G_1$ and $G_2$. Continue splitting $G_1, G_2$, and so on, recursively, until no further split operation is possible. The resulting graphs are \textit{split components} of $G$. They can either be $K_3$ (triangles), triple bonds, or triconnected normal graphs.

Let $e_\mathcal{V}$ be a virtual edge. There are exactly two split components containing $e_\mathcal{V}$: $G_1 = (V_1, E_1)$ and $G_2 = (V_2, E_2)$. Replacing $G_1$ and $G_2$ with $G' = (V_1 \cup V_2, (E_1 \cup E_2) \setminus \{ e_\mathcal{V} \})$ is called \textit{merging} $G_1$ and $G_2$. Do all possible mergings of the cycle graphs (starting from triangles), and then do all possible mergings of multiple bonds starting from triple bonds. Components of the resulting set are referred to as the \textit{triconnected components} of $G$. We emphasize again that some graphs (i.e., cycles and bonds) in the set of triconnected components are not necessarily triconnected.

\begin{lemma}
\cite{hopcroft2} Triconnected components are unique for $G$. Total number of edges within the triconnected components is at most $3 | E | - 6$.
\label{lemma:3}
\end{lemma}

Consider a graph $T$, where vertices (further referred to as \textit{nodes} for disambiguation) are triconnected components, and there is an edge between $a$ and $b$ in $T$, when $a$ and $b$ share a (copied) virtual edge.

\begin{lemma}
\cite{hopcroft2} $T$ is a tree.
\end{lemma}

\paragraph{Example.} Figure~\ref{fig:gr_tree} illustrates triconnected decomposition of a binconnected graph and intermediate steps towards it.

All triconnected components, and thus $T$, can be found in $O(N + | E |)$ steps \cite{hopcroft2,gutwenger,vo}. Merging of two triconnected components is equivalent to contracting an edge in $T$ (VI on Figure \ref{fig:gr_tree}). 
After all possible mergings, $G$ is recovered.

\subsection{Inference via Dynamic Programming} \label{subsec:inf}

Assume that there is a (small) number $C$ bounding the size of each nonplanar triconnected component. In the following, we present a polynomial time algorithm that computes $Z$ for a given (fixed) $C$. 

First, one finds triconnected components of $G$ and $T$ in $O(N + | E |)$ steps.
Choose a \textit{root} node $d$ in $T$. For any node $a \neq d$ in $T$, let the next node $b$ (on a unique path from $a$ to $d$) be a \textit{parent} of $a$, and $a$ be a \textit{child} of $b$. Nodes, which do not have any children, are called \textit{leaves}. For node $a$, let a \textit{subtree} $T(a)$ denote a subgraph constructed from $a$, its children, grandchildren, and so on.

Our algorithm processes each node once. The node is only processed when all its children have been already processed, so a leaf is processed first and the root is processed last. Let $a = (V_a, E_a), N_a = | V_a |$ be a currently processed node. Let $G^T_a = (V^T_a, E^T_a)$ be a graph obtained by merging all nodes in $T(a)$. If $a$ is a root, then $G^T_a = G$. Since the root is processed last, it outputs the desired PF, Z. Figure \ref{fig:k33free_inf} provides a visualization of a node processing routine which is to be explained.

If $a$ is not a root, let $e_\mathcal{V} = \{ p, t \}$ be a virtual edge shared between $a$ and its parent. The only virtual edge in $G^T_a$ is $e_\mathcal{V}$, and $G^T_a$ without $e_\mathcal{V}$ is a subgraph of $G$. Hence, pairwise interactions are defined for $E^T_a \setminus \{ e_\mathcal{V} \}$. The result of node $a$'s processing is a quantity.
\begin{equation*}
    \pi_{a} (x', x'') = \sum_{\substack{x_p = x', x_t = x'' \\ \forall u \in V^T_a \setminus e_\mathcal{V} : \, x_u = \pm 1}} \exp \biggl( \sum_{\substack{e = \{ v, w \} \\ e \in E^T_a \setminus \{ e_\mathcal{V} \}}} J_e x_v x_w \biggr),
\end{equation*}
where $x',x'' = \pm 1$. Notice that $\pi_{a} (+1, +1) = \pi_{a} (-1, -1)$, $\pi_{a} (+1, -1) = \pi_{a} (-1, +1)$, and hence $\pi_a(x', x'') = \pi_a(x'', x')$.

Processing nodes one by one we notice that the following cases are possible:
\begin{enumerate}[wide, labelwidth=!, labelindent=0pt]
\item \textbf{$a$ is a leaf}. Therefore, there is nothing to merge, and $a = G^T_a = (V_a, E_a)$. If $a$ is nonplanar, find $\pi_{a} (\pm 1, \pm 1)$ by brute force enumeration, completed in $O(1)$ steps. If $a$ is a multiple bond, $\pi_a(\pm 1, \pm 1)$ is found in $O(| E_a |)$ steps.

Assume now that node $a$ is (or corresponds to) a planar, normal graph. Define $J_{e_\mathcal{V}} = 0$ and consider a ZFI model with the probability $\mathbb{P}_a (S_a = X_a)$ defined over graph $a$ with $\{ J_e \, | \, e \in E_a \}$ as pairwise interactions. Let $Z_a$ be the PF of the ZFI model. In the remaining part of this case we will only work with this induced ZFI model, so that one can assume that nodes in $V_a$ are ordered, $V_a = \{ v_1, ..., v_{N_a} \}$, such that $v_1 = p, v_2 = t$. Then, one utilizes the notations $S_a = (s_1, ..., s_{N_a})$ and $X_a = (x_1, ..., x_{N_a}) \in \{ -1, +1 \}^{N_a}$ and derives
\begin{align}
\pi_{a} (x', x'') &= \sum_{X_a = (x', x'', \pm 1, \dots, \pm 1)} \exp \biggl( \sum_{\substack{e = \{ v, w \} \\ e \in E_a}} J_e x_v x_w \biggr) \nonumber\\
&= Z_a \mathbb{P}_a (x_1 = x', x_2 = x'')\label{eq:pi}.
\end{align}
Next, one triangulates $a$ by adding enough edges with zero pairwise-interactions, similar to how it is done in Subsection \ref{subsec:edg}. Assume that $a$ is triangulated, and observe that the right-hand side of Eq.~(\ref{eq:pi}) is not affected. Construct $G^* = (V^*, E^*)$, which is an expanded dual graph of $a$ with $E^*_I, E^*_C$, and $g$ defined as in Subsection \ref{subsec:edg}. Then, define mapping $M: \{ -1, +1 \}^{N_a} \to \text{PM}(G^*)$, weights $c_{e^*}$, and the PF $Z^*$ as in \ref{subsec:pmc}. Denote $e^*_\mathcal{V} = g^{-1} (e_\mathcal{V})$.

According to the definition of $M$,
\begin{align*}
\mathbb{P}_a (x_1 = x_2) &= \mathbb{P}_a (e^*_\mathcal{V} \in M(S_a)) \\
&= \frac{1}{Z^*} \sum_{\substack{E' \in \text{PM}(G^*), \\ \, e^*_\mathcal{V} \in E'}}  \prod_{e^* \in E'} c_{e^*}. 
\numberthis
\label{eq:eqprob}
\end{align*}
Denote $G^*_\mathcal{V} = G^*(V^* \setminus e^*_\mathcal{V})$. We continue the chain of relations/equalities (\ref{eq:eqprob}) observing that
\begin{align*}
\{ E' \in \text{PM}(G^*) \, | \, e^*_\mathcal{V} \in E' \} = \{ E'' \cup \{ e^*_\mathcal{V} \} \, | \, E'' \in \text{PM}(G^*_\mathcal{V}) \}.
\end{align*}
Then one arrives at
\begin{align*}
\mathbb{P}_a (x_1 = x_2) &= \frac{c_{e^*_\mathcal{V}}}{Z^*} \sum_{E'' \in \text{PM}(G^*_\mathcal{V})} \prod_{e^* \in E''} c_{e^*} = \frac{c_{e^*_\mathcal{V}} Z^*_\mathcal{V}}{Z^*},
\end{align*}
where $Z^*_\mathcal{V}$ is a PF of the PM model over $G^*_\mathcal{V}$. Compute $Z^*$ and  $Z_a$ in $O(N_a^\frac32)$ steps, as described in Section \ref{sec:planar}. Since $G^*_\mathcal{V}$ is planar of size $O(N_a)$, $Z^*_\mathcal{V}$ can also be computed in $O(N_a^\frac32)$ steps, as Theorem \ref{th:pmmodel} states. The following relations finalize computation of $\pi_a (\pm 1, \pm 1)$ in $O(N_a^\frac32)$ steps:
\begin{align*}
\pi_a (+1, +1) &=
\frac{Z_a}{2} \mathbb{P}_a (x_1 = x_2) = \frac{Z_a e^*_\mathcal{V} Z^*_\mathcal{V}}{2 Z^*} \\
\pi_a (+1, -1) &=
\frac{Z_a}{2} \mathbb{P}_a (x_1 \neq x_2) = \frac{Z_a}{2} - \pi_a (+1, +1).
\end{align*}

\item \textbf{$a$ is not a leaf, not a root}. Let $c_1, ..., c_q$ be $a$'s children, and $e^i_\mathcal{V} = \{ p^i, t^i \}$ be a virtual edge shared between $c_i$ and $a$, $1 \leq i \leq q$. At this point, we already computed all $\pi_{c_i} (\pm 1, \pm 1)$. Each $\{ p^i, t^i \}$ is a separation pair in $G^T_a$ that splits it into $G^T_{c_i}$ and the rest of $G^T_a$, containing all $G^T_{c_j}$, $j \neq i$. Denote all virtual edges in $a$ as $E_\mathcal{V}$, and then the following relation  holds:
\begin{eqnarray}
\pi_a (x', x'') = \sum_{\substack{x_p = x', x_t = x'', \\ \forall u \in V_a \setminus e_\mathcal{V}: x_u = \pm 1}} \biggl[ \exp\biggl( \sum_{\substack{e = \{ v, w \}\\ e \in E_a \setminus E_\mathcal{V}}} J_e x_v x_w\biggr) \nonumber\\
\cdot \prod_{i = 1}^q \pi_{c_i} (x_{p^i}, x_{t^i}) \biggr].
\label{eq:bellman}
\end{eqnarray}

If $a$ is (or corresponds to) a multiple bond, (\ref{eq:bellman}) is computed trivially in $O(| E_a |)$ steps. Hence, one assumes next that $a$ is a normal graph.

Each $\pi_{c_i} (x', x'')$ is positive, and it essentially only depends on the product $x' x''$, that is, there exist such $A_i, B_i$ that $\log \pi_{c_i} (x', x'') = A_i + B_i x' x''$. Using this relation, one rewrites (\ref{eq:bellman}) as
\begin{eqnarray}
\pi_a (x', x'') = \sum_{\substack{x_p = x', x_t = x'', \\ \forall u \in V_a \setminus e_\mathcal{V}: x_u = \pm 1}} \exp\biggl(\sum_{\substack{e = \{ v, w \} \\ e \in E_a \setminus E_\mathcal{V}}} J_e x_v x_w \nonumber\\
+ \sum_{i = 1}^q B_i x_{p^i} x_{t^i}\biggr) \cdot \exp\biggl( \sum_{i = 1}^q A_i\biggr).
\label{eq:bellman2}
\end{eqnarray}
Denote $J_{e_\mathcal{V}} = 0$, $J_{e^i_\mathcal{V}} = B_i$ for each $1 \leq i \leq q$. Then rewrite (\ref{eq:bellman2}) as
\begin{align}
\pi_a &(x', x'') =  \exp\biggl( \sum_{i = 1}^q A_i \biggr) \nonumber\\
&\cdot  \sum_{\substack{x_p = x', x_t = x'', \\ \forall u \in V_a \setminus e_\mathcal{V}: x_u = \pm 1}} 
\exp\biggl( \sum_{e = \{ v, w \} \in E_a} J_e x_v x_w\biggr). \label{eq:bellman3}
\end{align}
We compute (\ref{eq:bellman3}) by brute force in $O(1)$ steps, if $a$ is nonplanar. If $a$ is normal planar, we once again consider a ZFI model with the probability $\mathbb{P}_a (S_a = X_a)$, defined over $G_a$, where the pairwise weights are $\{ J_e \, | \, e \in E_a \}$, and $Z_a$ is the respective PF. Then applying machinery from Case 1, one derives 
\begin{equation*}
\pi_a (x', x'') = \exp\biggl( \sum_{i = 1}^q A_i \biggr) \cdot Z_a \mathbb{P}_a (x_p = x', x_t = x'')
\end{equation*}
in $O(N_a^\frac32)$ steps.

\item \textbf{$a$ is a root}. Once again, let $c_1, ..., c_q$ be children of $a$, $e^i_\mathcal{V} = \{ p^i, t^i \}$ be a virtual edge shared between $c_i$ and $a$, and $1 \leq i \leq q$, $E_\mathcal{V}$ be the set of virtual edges in $E_a$ (which $a$ shares only with its children). Using considerations similar to those described while deriving Eq.~(\ref{eq:bellman}), one arrives at
\begin{align*}
&Z = \sum_{X \in \{ -1, +1 \}^N} \exp \biggl( \sum_{e = \{v, w\} \in E} J_e x_v x_w \biggr) = \\
&\sum_{\forall u \in V_a: x_u = \pm 1} \biggl[ \exp \biggl( \sum_{\substack{e = \{ v, w \}\\ e\in E_a \setminus E_\mathcal{V}}} J_e x_v x_w \biggr) \cdot \prod_{i = 1}^q \pi_{c_i} (x_{p^i}, x_{t^i}) \biggr].
\end{align*}

Finally, one computes $Z$ similarly to how the $\pi$ values were derived in Case 2. It takes $O(| E_a |)$ steps if $a$ is a multiple bond. Otherwise, one constructs a ZFI model and finds the PF over the respective graphs in either $O(1)$ steps, if the graph is nonplanar, or in $O(N_a^\frac32)$ steps, if $a$ is normal planar.

\end{enumerate}

\begin{figure*}[!h] \centering
\centering
\includegraphics[width=0.95\linewidth]{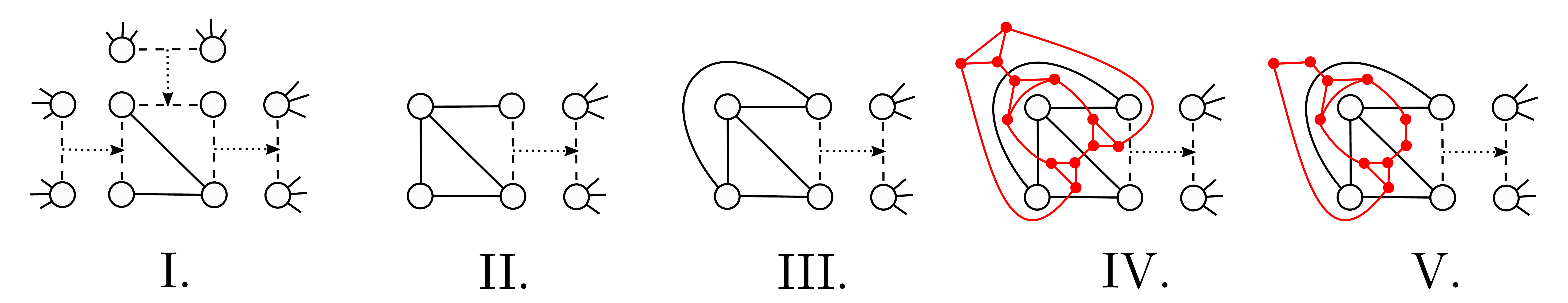}
\caption{Inference. Illustration of a node processing. Arrow indicates a direction to the root. (I) Exemplary node $a$ (subgraph in the center with one solid side edge, one solid diagonal edge, and solid dashed edges, marked according to the rules explained in the captions to Fig.~\ref{fig:gr_tree}), its (two) children and a parent. (II) Topology of the ZFI model defined on $a$. (III) Triangulated ZFI model. (IV) Expanded dual graph $G^*$ of ZFI model (red). Computing PF $Z^*$ of $G^*$'s PMs is a part of the inference processing of the node $a$. (V) $G^*_\mathcal{V}$ graph for $a$ (red). Computing PF $Z^*_\mathcal{V}$ of $G^*_\mathcal{V}$ PMs is a part of the inference processing of node $a$,  unless $a$ is a root.}
\label{fig:k33free_inf}
\end{figure*}

\begin{figure*}[!h] \centering
\centering
\includegraphics[width=0.95\linewidth]{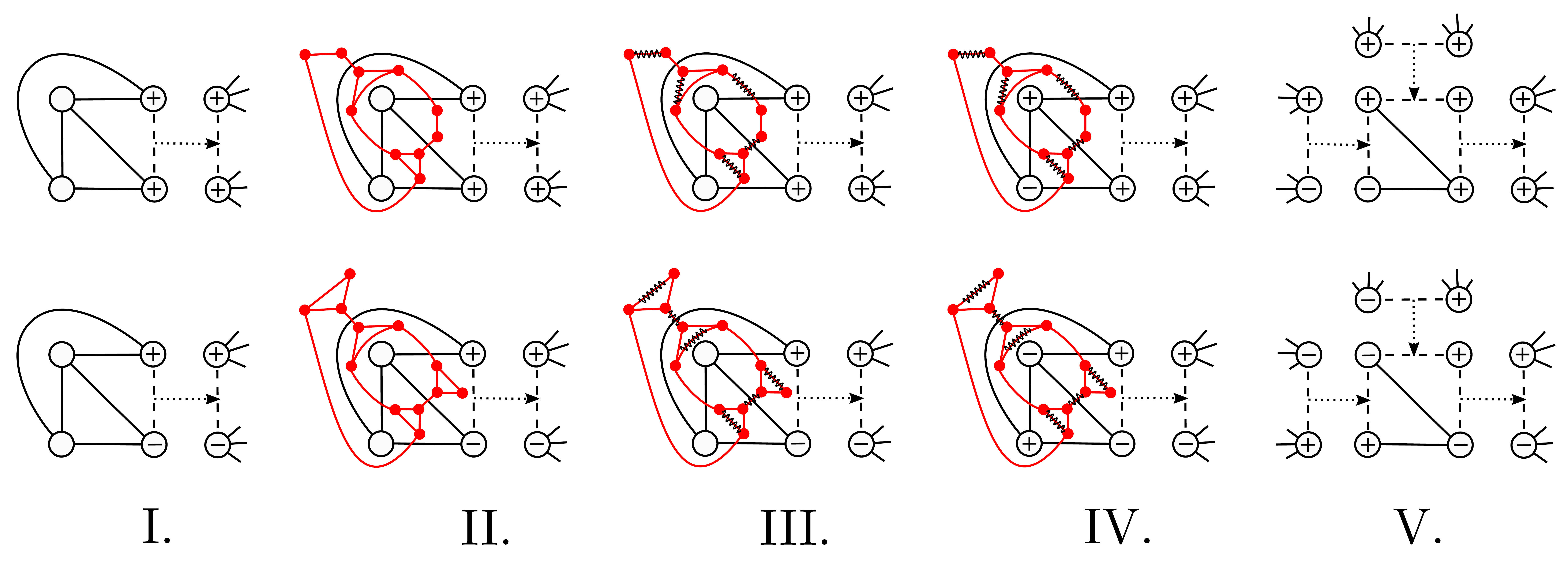}
\caption{Sampling. Illustration of a node processing. General notations (arrows, children, parents, dashed and dotted lines) are consistent with the captions of Figs.~\ref{fig:gr_tree},\ref{fig:k33free_inf}. Assume that spin values at $a$'s parent are already drawn (and consequently, spin values at $e_\mathcal{V}$ are drawn, too). The examples in the top line are for the case of equal spin values at $e_\mathcal{V}$, and the examples in the bottom line are for unequal spin values at $e_\mathcal{V}$. (I) Start with the triangulated ZFI model defined during inference (see Fig.~\ref{fig:k33free_inf}). (II) Find either $G^*_\mathcal{V}$ (top, red) or $\overline{G}^*_\mathcal{V}$ (bottom, red) depending on spin values at $e_\mathcal{V}$. (III) Sample PM on $G^*_\mathcal{V}$ or $\overline{G}^*_\mathcal{V}$. (IV) Set spin values according to PM. (V) Propagate the spin values drawn along the virtual edges towards the child nodes.}
\label{fig:k33free_samp}
\end{figure*}

\begin{figure}[!h]
\centering
  \includegraphics[width=0.5\linewidth]{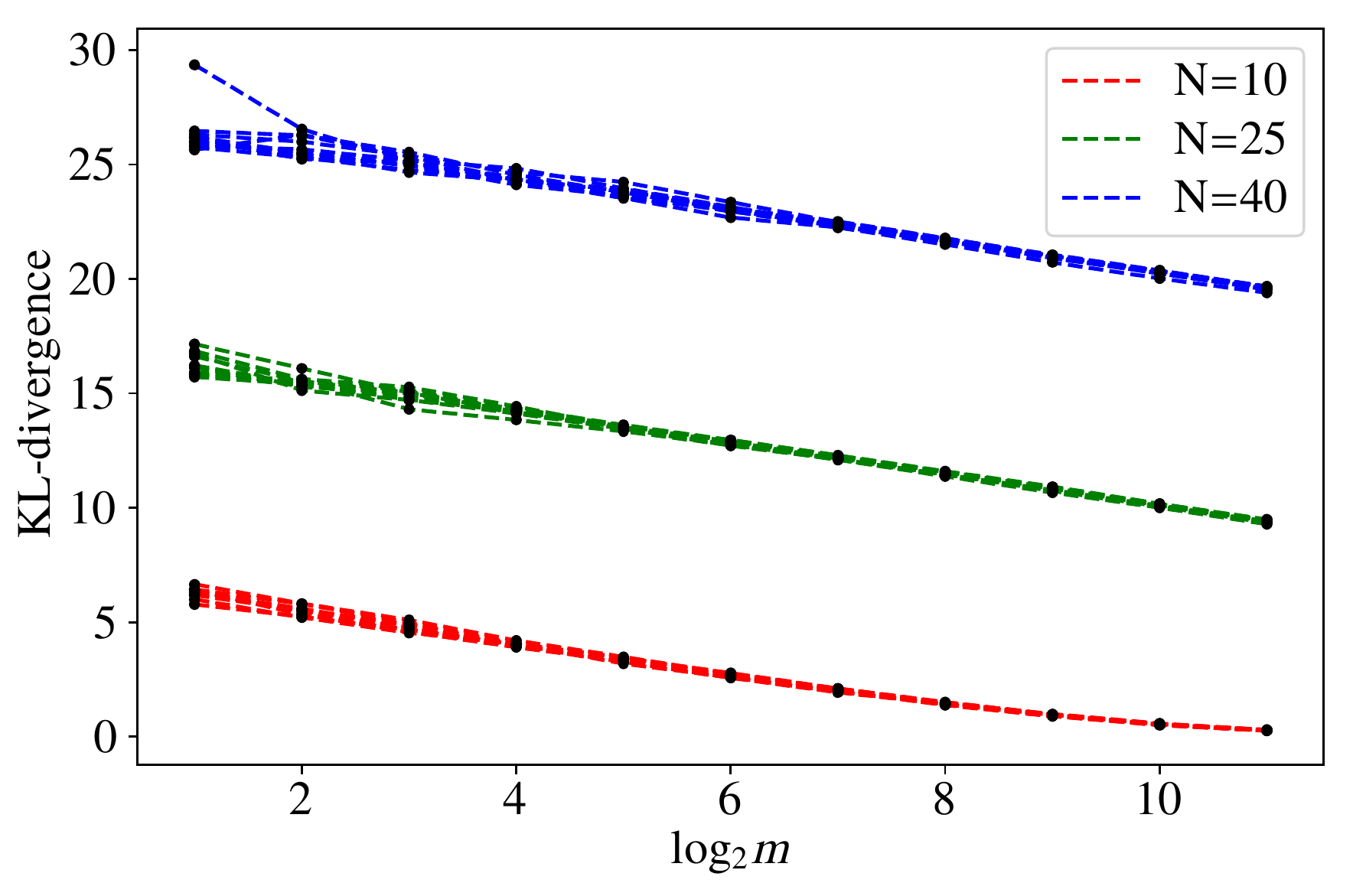}
  \caption{KL-distance of the model probability distribution compared with the empirical probability distribution. $N, m$ are the model's size and the number of samples, respectively.}
  \label{fig:kl}
\end{figure}

\begin{figure}[!h]
\centering
  \includegraphics[width=0.5\linewidth]{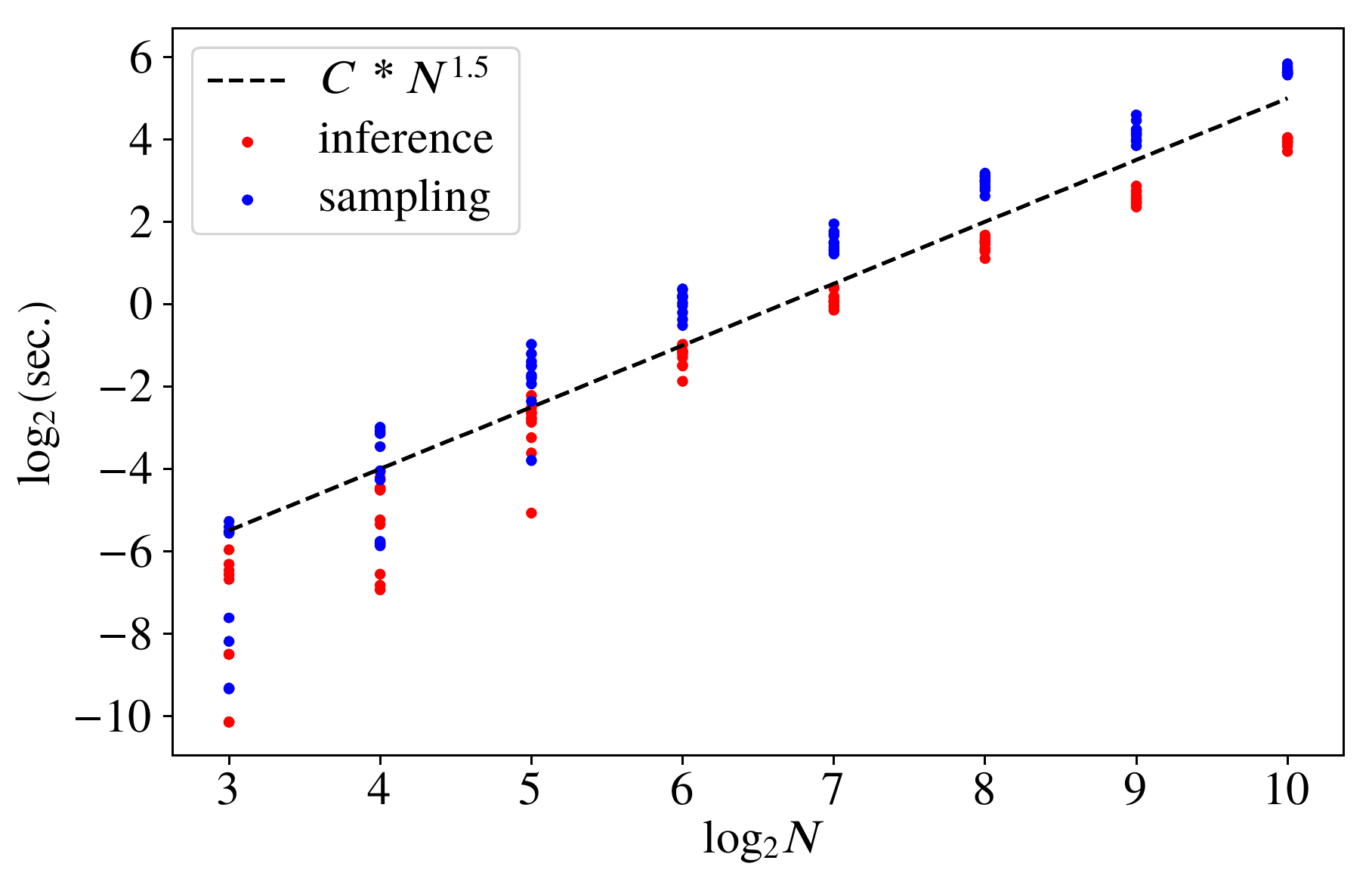}
  \caption{Execution time of inference (red dots) and sampling (blue dots) depending on $N$, shown on a logarithmic scale. Black line corresponds to $O(N^\frac32)$.}
  \label{fig:time}
\end{figure}

\subsection{Sampling via Dynamic Programming} \label{subsec:sample}

The sampling algorithm, detailed below, follows naturally from the inference routine. Compute triconnected components of $G$ in $O(N + | E |)$ steps. If all the triconnected components of $G$ are multiple bonds, $G$ should be a multiple bond itself, but $G$ is normal. Therefore, there exists a component that is not a multiple bond; choose it as a root of $T$.

Use the inference routine (described in the previous Section) to compute $Z$. Now, do a backward pass through the tree, processing the root first, and then processing the node only when its parent has already been processed (Figure \ref{fig:k33free_samp} visualizes the sampling algorithm).

Suppose $a$ is a root and it is processed by now. Since $a$ is not a multiple bond, it results in an Ising model, $\mathbb{P}_a (S_a = X_a)$. Draw a spin configuration $X_a$ from this model. It will take $O(1)$ steps if $a$ is nonplanar or $O(N_a^\frac32)$ steps if $a$ is planar.

Suppose $a$ is not a root. If $a$ is a multiple bond, spin values were already assigned to its vertices (contained within the node/graph $a$). Otherwise, there exists a ZFI model $\mathbb{P}_a (S_a = X_a)$ already constructed at the inference stage. Following the notation of Subsection \ref{subsec:inf}, one has to sample from $\mathbb{P}_a (S_a = X_a | s_p = x_p, s_t = x_t)$, since spins $s_p$ and  $s_t$ are shared with the parent model and have already been drawn as $x_p$ and $x_t$, respectively. If $x_p = x_t$, all valid $X_a$ are such that $e^*_\mathcal{V} \in M(X_a)$, and the task is reduced to sampling PMs on $G^*_\mathcal{V}$. Otherwise, all valid $X_a$ are such that $e^*_\mathcal{V} \notin M(X_a)$. Denote $\overline{G}^*_\mathcal{V} = (V^*, E^* \setminus \{ e^*_\mathcal{V} \})$ and notice that
\begin{equation*}
\{ E' \in \text{PM}(G^*) \, | \, e^*_\mathcal{V} \notin E' \} = \text{PM}(\overline{G}^*_\mathcal{V}).
\end{equation*}
Therefore, the task is reduced to sampling PM over $\overline{G}^*_\mathcal{V}$.

\section{$K_{33}$-free Topology} \label{sec:k33}

\subsection{ZFI Model over $K_{33}$-free Graphs}

Consider the ZFI model (\ref{eq:distr}) over a normal connected graph $G$.
Let $H$ be some graph. Then, $H$ is a \textit{minor} of $G$, if it is isomorphic to $G$'s subgraph, in which some edges are contracted. (See \cite{diestel}, Chapter 1.7, for a formal definition.)

$G$ is \textit{$K_{33}$-free}, if $K_{33}$ is not a minor of $G$, that is, it cannot be derived from $G$'s subgraph by contraction of some edges.

Let a biconnected $G$ be decomposed into the tree of triconnected components. Then, the following lemma holds:
\begin{lemma} \label{th:k33free}
\cite{hall} Graph $G$ is $K_{33}$-free if and only if its nonplanar triconnected components are exactly~$K_5$.
\end{lemma}


Therefore, if $G$ is $K_{33}$-free, it satisfies all the conditions needed for efficient inference and sampling, described in Section \ref{sec:dynprog}. According to the lemma, the graph in Fig.~\ref{fig:gr_tree} is $K_{33}$-free. The next statement expresses the main contribution of this manuscript.

\begin{theorem} \label{th:k33comp}
If $G$ is $K_{33}$-free, inference or sampling of (\ref{eq:distr}) takes $O(N^\frac32)$ steps.
\end{theorem}

We point out that the family of models for which the algorithm from Section \ref{sec:dynprog} applies is broader than just $K_{33}$-free models. However, we focus on $K_{33}$-free graphs because they have a fortunate characterization in terms of a missing minor.

\subsection{Discussion: Genus of $K_{33}$-free Graphs}

A remarkable feature of $K_{33}$-free models is related to considerations addressing the graph's genus.  \textit{Genus} of a graph is a minimal genus (number of handles) of the orientable surface that the graph can be embedded into. Kasteleyn \cite{kasteleyn} has conjected that the complexity of evaluating the PF of a ZFI model embedded in a graph of genus $g$ is exponential in $g$. The result was proven and detailed in 
 \cite{regge,gallucio,cimasoni1,cimasoni2}.
One naturally asks what are genera of graphs over which the ZFI models are tractable.  The following statement relates biconnectivity and graph topology (genus):  
\begin{theorem}
\cite{battle} A graph's genus is a sum of its biconnected component genera. 
\end{theorem}
If a graph is not biconnected, its genus can be arbitrarily large, while inference and sampling may  still be tractable in relation to the decomposition technique discussed in Subsection \ref{subsec:dec}. Therefore, it becomes principally interesting to construct tractable biconnected models with large genus.

\begin{lemma} \label{lemma:k33ex}
    A biconnected $K_{33}$-free graph of size $5 n$ can be of genus as big as $n$.
\end{lemma}

From this we conclude that $K_{33}$-free graphs can't be tackled via the bounded-genus approach of \cite{regge,gallucio,cimasoni1,cimasoni2}. This justifies the novelty of our contribution.

\section{Implementation and Tests} \label{sec:imp}

To test the correctness of inference, we generate random $K_{33}$-free models of a given size and then compare the value of PF computed in a brute force way (tractable for sufficiently small graphs) and by our algorithm. We simulate samples of sizes from $\{ 10, ..., 15 \}$ ($1000$ samples per size) and verify that respective expressions coincide.

When testing sampling implementation, we take for granted that the produced samples do not correlate given that the sampling procedure (Section \ref{subsec:sample}) accepts the Ising model as input and uses independent random number generation inside. The construction does not have any memory, therefore, it generates statistically independent samples. To test that the empirical distribution is approaching a theoretical one (in the limit of the infinite number of samples), we draw different numbers, $m$, of samples from a model of size $N$. Then we find Kullback-Leibler divergence between the probability distribution of the model (here we use our inference algorithm to compute the normalization, $Z$) and the empirical probability, obtained from samples. Fig.~\ref{fig:kl} shows that KL-divergence converges to zero as the sample size increases. Zero KL-divergence corresponds to equal distributions.

Finally, we simulate inference and sampling for random models of different size $N$ and observe that the computational time (efforts) scales as $O(N^\frac32)$ (Fig.~\ref{fig:time})\footnote{Implementation of the algorithms is available at \href{https://github.com/ValeryTyumen/planar_ising}{https://github.com/ValeryTyumen/planar\_ising}.}.

\section{Conclusion}
\label{sec:conclusions}

In this manuscript, we compiled results that were scattered over the literature on $O(N^\frac32)$ sampling and inference in the Ising model over planar graphs. To the best of our knowledge,  we are the first to present a complete and mathematically accurate description of the tight asymptotic bounds.


We generalized the planar results to a new class of zero-field Ising models over graphs not containing $K_{33}$ as a minor. In this case, which is strictly more general than the planar case, we have shown that the complexity bounds for sampling and inference are the same as in the planar case.
Along with the formal proof, we provided evidence of our algorithm's correctness and complexity through simulations.
\section*{Acknowledgements}
This work was supported by the U.S. Department of Energy through the Los Alamos National Laboratory as part of LDRD and the DOE Grid Modernization Laboratory Consortium (GMLC). Los Alamos National Laboratory is operated by Triad National Security, LLC, for the National Nuclear Security Administration of U.S. Department of Energy (Contract No. 89233218CNA000001). 




\bibliography{sample}
\bibliographystyle{chicago}

\begin{appendices}

\section{Technical Proofs}

\textbf{Lemma \ref{lemma:bij} proof.} Let $E' \in \text{PM}(G^*)$. Call $e \in E$ \textit{saturated}, if it intersects an edge from $E' \cap E^*_I$. Each Fisher city is incident to an odd number of edges in $E' \cap E^*_I$. Thus, each face of $G$ has an even number of unsaturated edges. This property is preserved, when two faces/cycles are merged into one by evaluating respective symmetric difference. Therefore, one gets that any cycle in $G$ has an even number of unsaturated edges.

For each $i$ define $x_i := -1^{r_i}$, where $r_i$ is the number of unsaturated edges on the path connecting $v_1$ and $v_i$. The definition is consistent due to aforementioned cycle property. Now for each $e = \{ v, w \} \in E$, $x_v = x_w$ if and only if $e$ is saturated. To conclude, we constructed $X$ such that $E' = M(X)$. Such $X$ is unique, because parity of unsaturated edges on a path between $v_1$ and $v_i$ uniquely determines relationship between $x_1$ and $x_i$, and $x_1$ is always $+1$.
\qed

\textbf{Lemma \ref{lemma:zfitopm} proof.} Let  $X' = (x'_1, ..., x'_N) \in \mathcal{C}_+$, $M(X') = E'$. The statement is justified by the following chain of transitions:
\begin{align*}
\mathbb{P} ( M(S) = E' )
&= \mathbb{P} ( S = X') + \mathbb{P} ( S = -X' ) \\
&= \frac2Z \exp \left(\sum_{e = \{ v, w \} \in E} J_e x'_v x'_w \right) \\
&= \frac2Z \exp \left(\sum_{e^* \in E' \cap E^*_I} 2 J_{g(e^*)} - \sum_{e \in E} J_e\right) \\
&= \frac2Z \exp \left(- \sum_{e \in E} J_e\right) \prod_{e^* \in E' \cap E^*_I} c_{e^*} \\
&= \frac2Z \exp \left(- \sum_{e \in E} J_e\right) \prod_{e^* \in E'} c_{e^*} \\
&= \frac{1}{Z^*} \prod_{e^* \in E'} c_{e^*}
\end{align*}

\begin{algorithm}
    \caption{Sampling from $\mathbb{P}(S = X)$}
    \label{alg:bisamp}
\begin{algorithmic}[1]
    \State {\bfseries Input:} A tree of $G_1, ..., G_h$.
    \State Draw $X_1, ..., X_h$ from ZFI models on $G_1, ..., G_h$.
    \State ProcessComponent(1, -1).
    \State Combine $X_1, ..., X_h$ into $X$.
    \State {\bfseries Output:} $X$.
    
    \State
    \State {\bfseries Procedure}{ ProcessComponent}
    \State {\bfseries Input:} index $i$, parent index $p$.
    \State $v$ = articulation point of $G_i$ and $G_p$.
    \If{unequal spins of $X_i$ and $X_p$ at $v$}
    \State $X_i$ := $-X_i$
    \EndIf
    \For{$j$ {\bfseries in} $i$'s neighbors}
    \If{$j \neq p$}
    \State ProcessComponent($j$, $i$).
    \EndIf
    \EndFor
    
    \State {\bfseries end Procedure}
\end{algorithmic}
\end{algorithm}

\textbf{Lemma \ref{lemma:bic} proof.} The Algorithm \ref{alg:bisamp} reduces sampling on $G$ to a series of samplings on $G_1, ..., G_h$.
 
Given the algorithm and inference formula in Lemma \ref{lemma:bic}, the statement is obvious for $h = 1$. Let $h = 2$. Let $v$ be an articulation point shared by $G_1$ and $G_2$. Denote $G_1 = (V_1, E_1)$, $G_2 = (V_2, E_2)$. Without loss of generality assume that $v$ has index $1$ in $V, V_1$ and $V_2$. Let $\mathcal{C}^i_+ = \{ +1 \} \times \{ -1, +1 \}^{| V_i |}$. Then one derives:
\begin{align*}
Z &= 2 \sum_{X \in \mathcal{C}_+} \exp\left( \sum_{e = \{v, w\} \in E} J_e x_v x_w\right) \\
&= 2 \sum_{X \in \mathcal{C}_+} \biggl[ \exp\left(\sum_{e = \{ v, w \} \in E_1} J_e x_v x_w\right) \cdot \exp\left(\sum_{e = \{ v, w \} \in E_2} J_e x_v x_w\right) \biggr] \\
&= 2 \sum_{X_1 \in \mathcal{C}^1_+} \exp\left(\sum_{e = \{ v, w \} \in E_1} J_e x_v x_w\right) \cdot \sum_{X_2 \in \mathcal{C}^2_+} \exp\left(\sum_{e = \{ v, w \} \in E_2} J_e x_v x_w\right) \\
&= \frac12 Z_1 Z_2
\end{align*}
where $Z_i$ is the PF of the ZFI model induced by $G_i$. As far as sampling is concerned, denote by $\mathbb{P}_i (S_i = X_i)$ a probability distribution induced by the $i$-th ZFI model. Then, since $\mathbb{P}_2 (s_1 = x_1) = \frac12$:
\begin{align*}
\mathbb{P}(S = X) &= \frac{1}{Z} \sum_{X \in \mathcal{C}_+} \exp\left( \sum_{e = \{v, w\} \in E} J_e x_v x_w\right) \\
&= 2  \frac{1}{Z_1} \exp\left(\sum_{e = \{ v, w \} \in E_1} J_e x_v x_w \right) \cdot \frac{1}{Z_2} \exp\left(\sum_{e = \{ v, w \} \in E_2} J_e x_v x_w\right) \\
&= 2 \mathbb{P}_1 (S_1 = X_1) \mathbb{P}_2 (S_2 = X_2) \\
&= \mathbb{P}_1 (S_1 = X_1) \frac{\mathbb{P}_2 (S_2 = X_2)}{\mathbb{P}_2 (s_1 = x_1)} \\
&= \mathbb{P}_1 (S_1 = X_1) \mathbb{P}_2 (S_2 = X_2 | s_1 = x_1)
\end{align*}

Assume that a method for sampling $S_i$ from $\mathbb{P}_i$ is available. Then, draw $X_1$ by sampling $S_1$ from $\mathbb{P}_1$. To sample $S_2$ conditional on $s_1 = x_1$ from $\mathbb{P}_2$, draw $X'_2 = (x'_1, ...)$ from $\mathbb{P}_2 (S_2 = X'_2)$. If $x'_1 = x_1$, then $X_2 = X'_2$, otherwise $X_2 = - X'_2$. This is consistent with Algorithm \ref{alg:bisamp}.

For graphs of $h > 2$ the statement of lemma follows naturally by induction.

\textbf{Theorem \ref{th:k33comp} proof.} Since $G$ is normal and minor-free, it holds that $| E | = O(N)$ \cite{thomason}. Find all biconnected components and for each construct a triconnected component tree in $O(N + | E |) = O(N)$.

As described above, the time (number of steps) of inference or sampling is a sum of inference or sampling times of each triconnected component of $G$. Let the set of all $G$'s triconnected components (that is, a union over all biconnected components) to consist of $k_1$ planar triconnected components of size $N_1, ..., N_{k_1}$ with $M^p_1, ..., M^p_{k_1}$ edges respectively, $k_2$ multiple bonds of $M^b_1, ..., M^b_{k_2}$ edges and $k_3$ $K_5$ graphs. Then the complexity of inference or sampling is $O(\sum_{i = 1}^{k_1} N_i^\frac32 + \sum_{i = 1}^{k_2} M^b_i + k_3)$.

The edges of $G$ are partitioned among biconnected components. Inside each biconnected component apply second part of Lemma \ref{lemma:3} to obtain that $\sum_{i = 1}^{k_1} M^p_i + \sum_{i = 1}^{k_2} M^b_i + 10 k_3 = O(| E |) = O(N)$. This gives that $\sum_{i = 1}^{k_2} M^b_i + k_3 = O(N)$ and $\sum_{i = 1}^{k_1} M^p_i = O(N)$. Since triconnected components are connected graphs, we get that $N_i = O(M^p_i)$ for all $1 \leq i \leq k_1$ and hence $\sum_{i = 1}^{k_1} N_i = O(N)$. From convexity of $f(x) = x^\frac32$ it follows that $\sum_{i = 1}^{k_1} N_i^\frac32 = O(N^\frac32)$ and finally that $O(\sum_{i = 1}^{k_1} N_i^\frac32 + \sum_{i = 1}^{k_2} M^b_i + k_3) = O(N^\frac32)$. \qed

\textbf{Lemma \ref{lemma:k33ex} proof.} 
A simple example illustrates that  genus of a biconnected $K_{33}$-free graph can grow linearly with its size. First, notice that $K_5$ is a nonplanar graph, but it can be embedded in toroid (Fig.~\ref{fig:k5emb}),
therefore genus of the graph is unity. Consider a cycle of length $2 n$, enumerate edges in the order of cycle traversal from $1$ to $2 n$. Attach $K_5$ graph to each odd edge of the cycle (see Fig.~\ref{fig:unbgen}). The resulting graph $G$ is of size $5 n$, it is biconnected and $K_{33}$-free (see Figure \ref{fig:triconn}). Remove an arbitrary even edge from the cycle. It results in a graph whose biconnected components are $n$ $K_5$ graphs and $n$ edges, so its genus is $n$. Since edge removal can only decrease genus, we conclude that $G$'s genus is at least $n$.

\begin{figure}[!h]
\centering
\subfigure[]{
  \centering
  \includegraphics[width=0.2\linewidth]{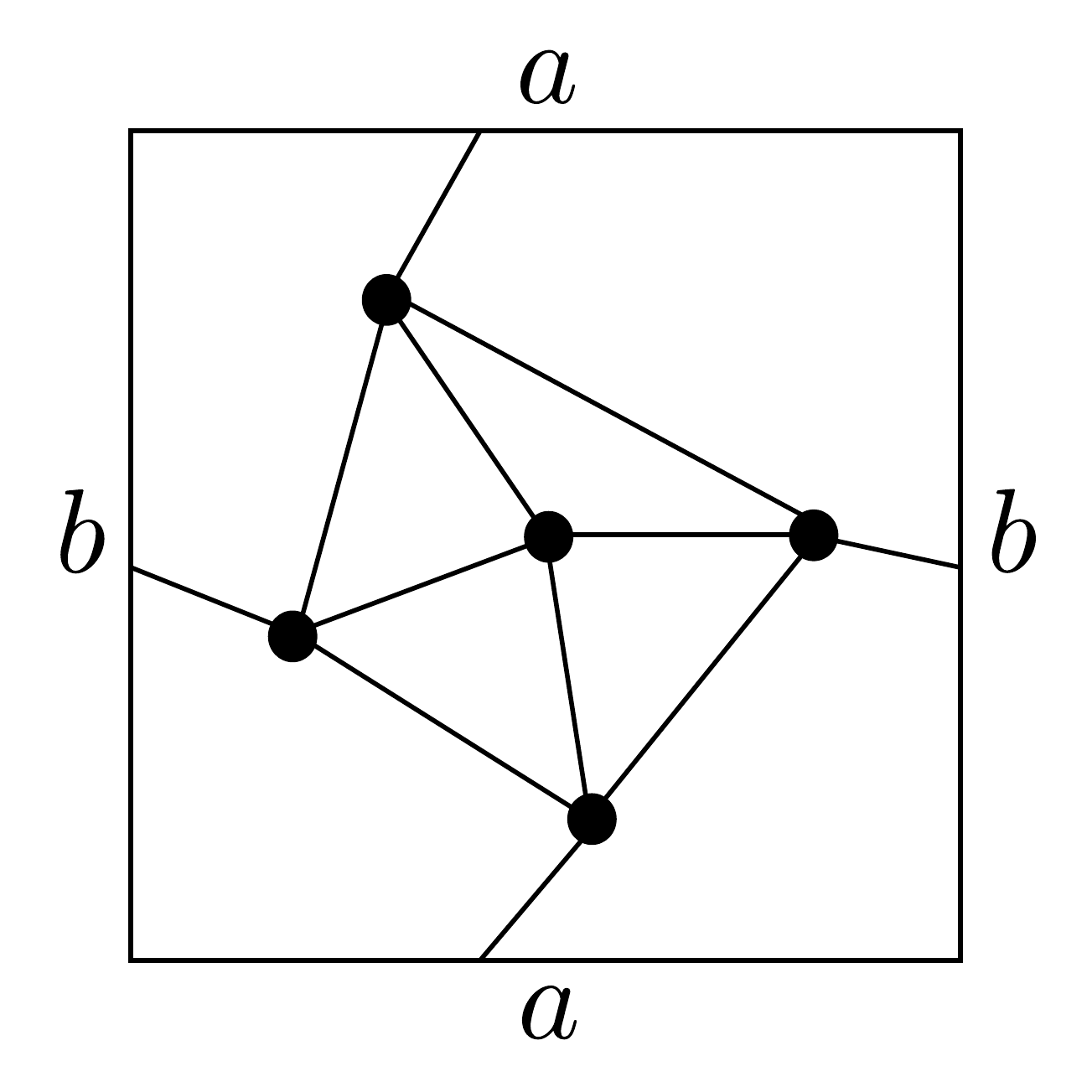}
  \label{fig:k5emb}
}%
\subfigure[]{
  \centering
  \includegraphics[width=0.3\linewidth]{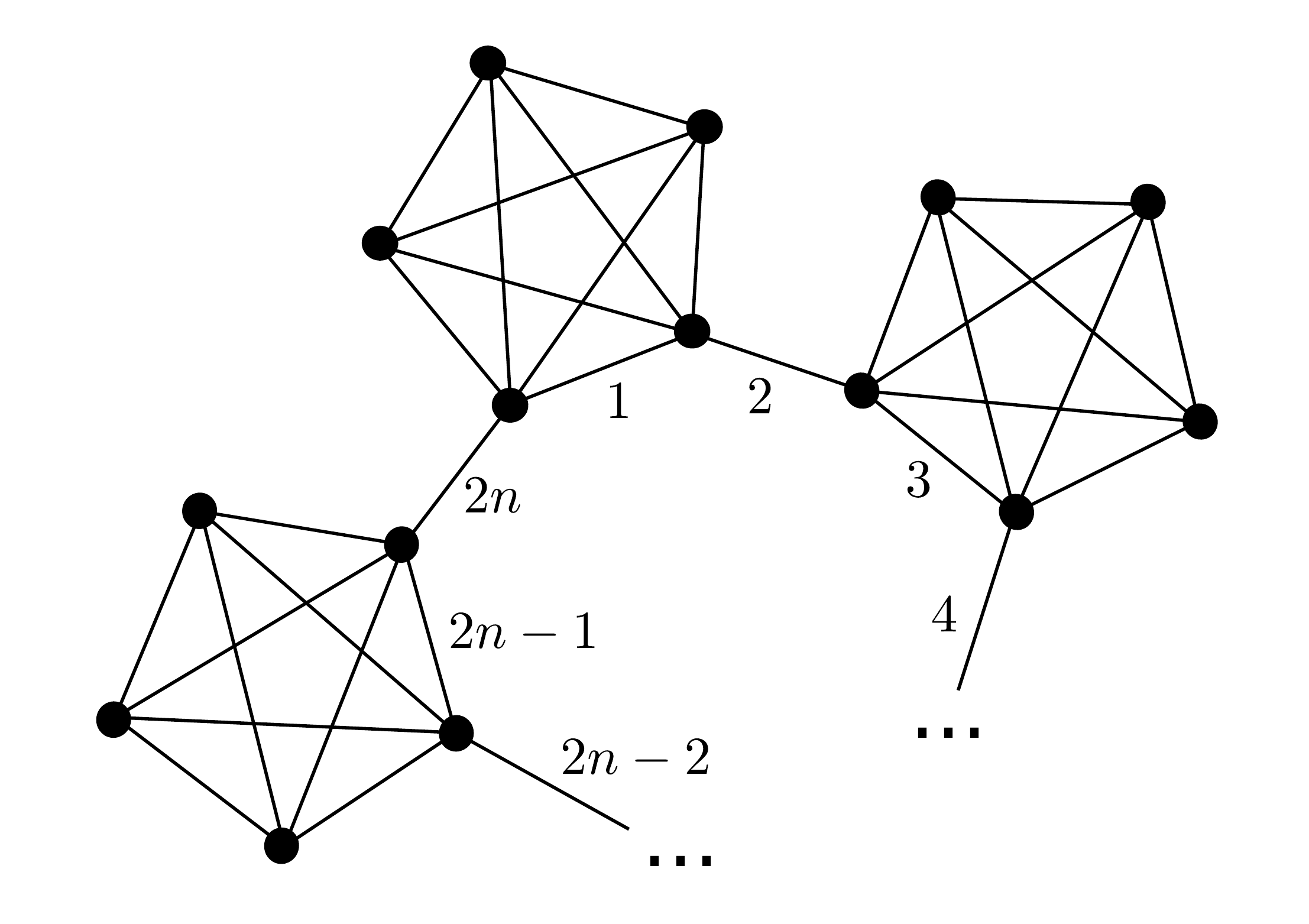}
  \label{fig:unbgen}
}%
\subfigure[]{
  \centering
  \includegraphics[width=0.4\linewidth]{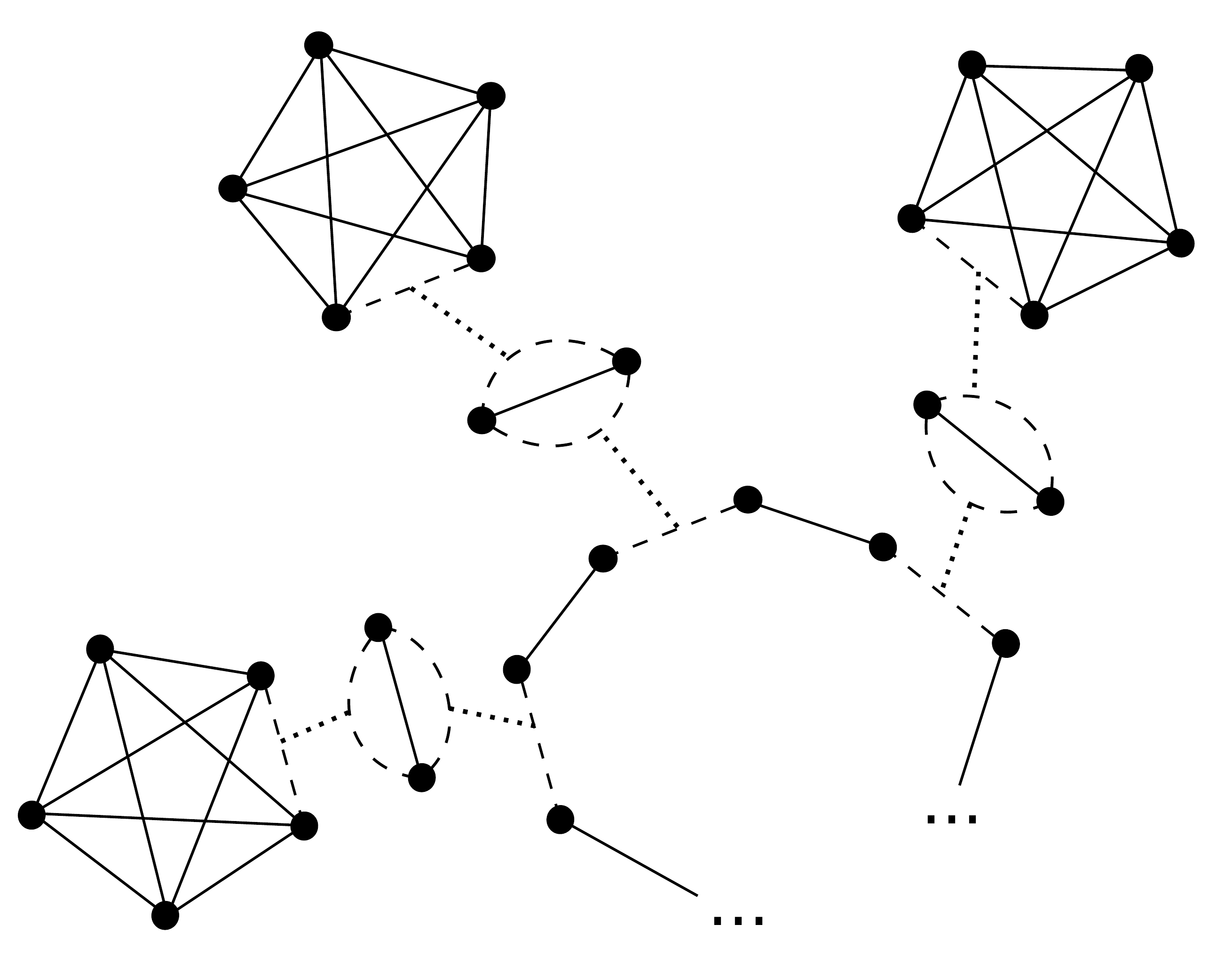}
  \label{fig:triconn}
}%
\caption{(a) $K_5$'s embedding on a toroid - glue sides with the same label together. (b) $G$ - a "necklace" of $n$ $K_5$ graphs. (c) $G$'s triconnected components. Dashed lines are virtual edges and dotted lines identify identical virtual edges. Triconnected components consist of a cycle, triple bonds and $K_5$ graphs. Hence, by Lemma \ref{th:k33free} $G$ is $K_{33}$-free.}
\end{figure}

\section{Counting PMs of Planar $\hat{G}$ in $O(\hat{N}^\frac32)$ time}

This section addresses inference part of Theorem \ref{th:pmmodel}.

\subsection{Pfaffian Orientation} \label{subseq:pf}

Let $\hat{G}$ be an oriented graph. Its cycle of even length (built on an even number of vertices) is said to be \textit{odd-oriented}, if, when all edges along the cycle are traversed in any direction, an odd number of edges are directed along the traversal. An orientation of $\hat{G}$ is called \textit{Pfaffian}, if all cycles $C$, such that $\text{PM}(\hat{G}(\hat{V} - C)) \neq \varnothing$, are odd-oriented.

We will need $\hat{G}$ to contain a Pfaffian orientation, moreover the construction is easy.
\begin{theorem}
Pfaffian orientation of $\hat{G}$ can be constructed in $O(\hat{N})$.
\end{theorem}
\begin{proof}
This theorem is proven constructively, see e.g. \cite{wilson,vazirani}, or \cite{schraudolph-kamenetsky}, where the latter construction is based on specifics of the expanded dual graph.
\end{proof}

Construct a skew-symmetric sparse matrix $K \in \mathbb{R}^{ \hat{N} \times \hat{N} }$ ($\to$ denotes orientation of edges):
\begin{equation}
K_{ij} = \begin{cases} c_e & \text{if } \{ v_i, v_j \} \in \hat{E}, v_i \to v_j \\ -c_e & \text{if } \{ v_i, v_j \} \in \hat{E}, v_j \to v_i \\ 0 & \text{if } \{ v_i, v_j \} \notin \hat{E} \end{cases}
\label{eq:kdef}
\end{equation}

The next result allows to compute PF $\hat{Z}$ of PM model on $\hat{G}$ in a polynomial time.
\begin{theorem}
$\det K > 0$, $\hat{Z} = \sqrt{\det K}$.
\end{theorem}
\begin{proof}
See, e.g., \cite{wilson} or \cite{kasteleyn}. 
\end{proof}

\subsection{Computing $\det K$}

LU-decomposition of a matrix $A = LU$, found via Gaussian elimination, where $L$ is a lower-triangular matrix with unit diagonals and $U$ is an upper-triangular matrix, would be a standard way of computing $\det A$, which is then equal to a product of the diagonal elements of $U$. However, this standard way of constructing the LU decomposition applies only if all $A$'s leading principal submatrices are nonsingular (See e.g. \cite{horn-johnson}, Section 3.5, for detailed discussions). And already the first, $1\times 1$,  leading principal submatrix of $K$ is zero/singular. 

Luckily, this difficulty can be resolved through the following construction. Take $\hat{G}$'s arbitrary perfect matching $E' \in \text{PM}(\hat{G})$. In the case of a general planar graph $E'$ can be found via e.g. Blum's algorithm \cite{blum} in $O(\sqrt{\hat{N}} | \hat{E} |) = O(\hat{N}^\frac32)$ time, while for graphs $G^*, G^*_v$ and $\overline{G}^*_v$ appearing in this paper $E'$ can be found in $O(N)$ from a spin configuration using $M$ mapping (e.g. $E' = E^*_I = M(\{ +1, ..., +1 \}) \in \text{PM}(G^*)$). Modify ordering of vertices, $\hat{V} = \{ v_1, v_2, ..., v_{\hat{N}} \}$, so that $E' = \{ \{ v_1, v_2 \}, ..., \{ v_{\hat{N} - 1}, v_{\hat{N}} \} \}$. Build $K$ according to the definition (\ref{eq:kdef}). Obtain $\overline{K}$ from $K$ by swapping column $1$ with column $2$, $3$ with $4$ and so on. This results in $\det K = | \det \overline{K} |$,  where the new $\overline{K}$ is properly conditioned.
\begin{lemma}
$\overline{K}$'s leading principal submatrices are nonsingular.
\end{lemma}
\begin{proof}
The proof, presented in \cite{wilson} for the case of unit weights~$c_e$, generalizes to arbitrary positive~$c_e$.
\end{proof}

Notice, that in the general case (of a matrix represented in terms of a general graph) complexity of the LU-decomposition is cubic in the size of the matrix. 
Fortunately, \textit{nested dissection} technique, discussed in the following subsection, allows to reduce complexity of computing $\hat{Z}$ to $O(\hat{N}^\frac32)$.

\subsection{Nested Dissection} \label{subsec:nd}

The partition $P_1, P_2, P_3$ of set $\hat{V}$ is a \textit{separation} of $\hat{G}$, if for any $v \in P_1, w \in P_2$ it holds that $\{ v, w \} \notin \hat{E}$. We refer to $P_1, P_2$ as the \textit{parts}, and to $P_3$ as the \textit{separator}.

Lipton and Tarjan (LT) \cite{lipton-tarjan} found an $O(\hat{N})$ algorithm, which finds a separation $P_1, P_2, P_3$ such that $\max ( | P_1 |, | P_2 | ) \leq \frac23 \hat{N}$ and $| P_3 | \leq 2^\frac32 \sqrt{ \hat{N} }$. The LT algorithm can be used to construct the so called \textit{nested dissection ordering} of $\hat{V}$. The ordering is built recursively, by first placing vertices of $P_1$, then $P_2$ and $P_3$, and finally permuting indices of $P_1$ and $P_2$ recursively according to the ordering of $\hat{G}(P_1)$ and $\hat{G}(P_2)$ (See \cite{lipton-rose-tarjan} for accurate description of details, definitions and analysis of the nested dissection ordering). As shown in \cite{lipton-rose-tarjan} the complexity of finding the nested dissection ordering is $O( \hat{N} \log \hat{N} )$.

Let $A$ be a $\hat{N} \times \hat{N}$ matrix with a \textit{sparsity pattern} of $\hat{G}$. That is, $A_{ij}$ can be nonzero only if $i = j$ or $\{ v_i, v_j \} \in \hat{E}$.
\begin{theorem}
\cite{lipton-rose-tarjan} If $\hat{V}$ is ordered according to the nested dissection and $A$'s leading principal submatrices are nonsingular, computing the LU-decomposition of $A$ becomes a problem of the $O(N^\frac32)$ complexity.
\end{theorem}
Notice, however, that we cannot directly apply the Theorem to $\overline{K}$, because the sparsity pattern of $K$ is asymmetric and does not correspond, in general, to any graph.

Let $G^{**} = (V^{**}, E^{**})$ be a planar graph, obtained from $\hat{G}$, by contracting each edge in $E'$, $| V^{**} | = | E' | = \frac12 \hat{N}$. Find and fix a nested dissection ordering  over $V^{**}$ (it takes $O(\hat{N} \log \hat{N})$ steps) and let the $\{ v_1, v_2 \}, \dots, \{ v_{\hat{N} - 1}, v_{\hat{N}} \}$ enumeration of $E'$ correspond to this ordering. Split $K$ into  $2 \times 2$ cells and consider the sparsity pattern of the nonzero cells. One observes that the resulting sparsity pattern coincides with the sparsity patterns of $\overline{K}$ and $G^{**}$.  Since LU-decomposition can be stated in the $2 \times 2$ block elimination form, its complexity is reduced down to $O(\hat{N}^\frac32)$.

This concludes construction of an efficient inference (counting) algorithm for planar PM model.

\section{Sampling PMs of Planar $\hat{G}$ in $O(\hat{N}^\frac32)$ time (Wilson's Algorithm)} \label{app:wilson}

This section addresses sampling part of Theorem \ref{th:pmmodel}. In this section we assume that degrees of $\hat{G}$'s vertices are upper-bounded by $3$.
This is true for $G^*$, $G^*_v$ and $\overline{G}^*_v$ - the only PM models appearing in the paper. Any other constant substituting $3$ wouldn't affect the analysis of complexity. Moreover, Wilson \cite{wilson} shows that any PM model on a planar graph can be reduced to bounded-degree planar model without affecting $O(\hat{N}^\frac32)$ complexity.

\subsection{Structure of the Algorithm}

Denote a sampled PM as $M$, $\mathbb{P}(M) = \hat{Z}^{-1} \prod_{e \in M} c_e$. Wilson's algorithm first applies LT algorithm of \cite{lipton-tarjan} to find a separation $P_1, P_2, P_3$ of $\hat{G}$ ($\max ( | P_1 |, | P_2 | ) \leq \frac23 \hat{N}$, $| P_3 | \leq 2^\frac32 \sqrt{ \hat{N} }$). Then it iterates over $v \in P_3$ and for each $v$ it draws an edge of $M$, saturating $v$. Then it appears that, given this intermediate result, drawing remaining edges of $M$ may be split into two independent drawings over $\hat{G}(P_1)$ and $\hat{G}(P_2)$, respectively, and then the process is repeated recursively. 

It takes $O(\hat{N}^\frac32)$ steps to sample edges attached to $P_3$ at the first step of the recursion, therefore the overall complexity of the Wilson's algorithm is also $O(\hat{N}^\frac32)$.

Subsection \ref{subsec:probs} introduces probabilities required to draw the aforementioned PM samples. Subsections \ref{subsec:step1} and \ref{subsec:step2} describe how to sample edges attached to the separator, while Subsection \ref{subsec:step3} focuses on describing the recursion.

\subsection{Drawing Perfect Matchings} \label{subsec:probs}

For some $Q \in \hat{E}$ consider the probability of getting $Q$ as a subset of $M$:
\begin{align*}
    \mathbb{P}(Q \subseteq M) &= \frac{1}{\hat{Z}} \sum_{\substack{M' \in \text{PM}(\hat{G}) \\ Q \subseteq M'}} \biggl( \prod_{e \in M'} c_e \biggr) \\
    &= \frac{1}{\hat{Z}} \biggl( \prod_{e \in Q} c_e \biggr) \cdot \sum_{M' \in \text{PM}(\hat{G})} \biggl( \prod_{e \in M' \setminus Q} c_e \biggr)
\numberthis \label{eq:qprob}
\end{align*}

Let $\hat{V}_Q = \cup_{e \in Q} e$ and $\hat{G}_{\setminus Q} = \hat{G} (\hat{V} \setminus \hat{V}_Q)$. Then the set $\{ M' \setminus Q \, | \, M' \in \text{PM}(\hat{G}) \}$ coincides with $\text{PM}(\hat{G}_{\setminus Q})$. This yields the following expression
\begin{equation*}
    \mathbb{P}(Q \subseteq M) = \frac{\hat{Z}_{\setminus Q}}{\hat{Z}} \biggl( \prod_{e \in Q} c_e \biggr)
\end{equation*}
where
\begin{equation*}
    \hat{Z}_{\setminus Q} = \sum_{M'' \in \text{PM}(\hat{G}_{\setminus Q})} \biggl( \prod_{e \in M''} c_e \biggr)
\end{equation*}
is a PF of the PM model on $\hat{G}_{\setminus Q}$ induced by the edge weights $c_e$.

For a square matrix $A$ let $A_{c_1, ..., c_l}^{r_1, ..., r_l}$ denote the matrix obtained by deleting rows $r_1, ..., r_l$ and columns $c_1, ..., c_l$ from $A$. Let $[A]_{c_1, ..., c_l}^{r_1, ..., r_l}$ be obtained by leaving only rows $r_1, ..., r_l$ and columns $c_1, ..., c_l$ of $A$ and placing them in this order.

Now let $\hat{V}_Q = \{ v_{i_1}, ..., v_{i_r} \}, i_1 < ... < i_r$. A simple check demonstrates that deleting vertex from a graph preserves the Pfaffian orientation. By induction this holds for any number of vertices deleted. From that it follows that $K_{i_1, ..., i_r}^{i_1, ..., i_r}$ is a Kasteleyn matrix for $\hat{G}_{\setminus Q}$ and then
\begin{equation*}
    \hat{Z}_{\setminus Q} = \pfaffian K_{i_1, ..., i_r}^{i_1, ..., i_r} = \sqrt{\det K_{i_1, ..., i_r}^{i_1, ..., i_r}}
\end{equation*}
resulting in
\begin{equation*}
    \mathbb{P}(Q \subseteq M) = \sqrt{ \frac{\det K_{i_1, ..., i_r}^{i_1, ..., i_r}}{\det K} } \cdot \biggl( \prod_{e \in Q} c_e \biggr)
\end{equation*}

Linear algebra transformations, described in \cite{wilson}, suggest that if $A$ is non-singular, then
\begin{equation*}
    \frac{\det A_{c_1, ..., c_l}^{r_1, ..., r_l}}{\det A} = \pm \det [A^{-1}]_{r_1, ..., r_l}^{c_1, ..., c_l}
\end{equation*}
This observation allows us to express probability (\ref{eq:qprob}) as
\begin{equation*}
    \mathbb{P}(Q \subseteq M) = \sqrt{ | \det [K^{-1}]_{i_1, ..., i_r}^{i_1, ..., i_r} |} \cdot \biggl( \prod_{e \in Q} c_e \biggr)
\end{equation*}

Now we are in the position to describe the first step of the Wilson's recursion.

\subsection{Step 1: Computing Lower-Right Submatrix of $\overline{K}^{-1}$} \label{subsec:step1}

Find a separation $P_1, P_2, P_3$ of $\hat{G}$. The goal is to sample an edge from every $v \in P_3$.

Let $T$ be a set of vertices from $P_3$ and their neighbors, then $| T | \leq 3 | P_3 |$ because each vertex in $\hat{G}$ is of degree at most $3$. Let $T^{**} \subseteq V^{**}$ be a set of the contracted edges (recall $G^{**}$ definition from Subsection \ref{subsec:nd}), containing at least one vertex from $T$, $| T^{**} | \leq | T |$. Then $T^{**}$ is a separator of $G^{**}$ such that
\begin{equation} \label{eq:tstarstar}
    | T^{**} | \leq | T | \leq 3 | P_3 | \leq 3 \cdot 2^\frac32 \sqrt{\hat{N}} \leq 3 \cdot 2^2 \sqrt{| V^{**} |}
\end{equation}
where one uses that, $| V^{**} | = \frac{\hat{N}}{2}$.
Find a nested dissection ordering (Subsection \ref{subsec:nd}) of $V^{**}$ with $T^{**}$ as a top-level separator. This is a correct nested dissection due to Eq.~(\ref{eq:tstarstar}).

Utilizing this ordering, construct $\overline{K}$. Compute $L$ and $U$ - LU-decomposition of $\overline{K}$ ($O(\hat{N}^\frac32)$ time).
Let $t = 2 |T^{**}| \leq 3 \cdot 2^\frac52 \sqrt{\hat{N}}$ and let $\mathcal{I}$ be a shorthand notation for $(\hat{N} - t + 1, ..., \hat{N})$. Using $L$ and $U$, find $D = [ \overline{K}^{-1} ]_\mathcal{I}^\mathcal{I}$, which is a lower-right $\overline{K}^{-1}$'s submatrix of size $t \times t$.

It is straightforward to observe that the $i$-th column of $D$, $d_i$, satisfies
\begin{equation*}
    [L]_\mathcal{I}^\mathcal{I} \times \biggl( [U]_\mathcal{I}^\mathcal{I} \times d_i \biggr) = e_i,
\end{equation*}
where $e_i$ is a zero vector with unity at the $i$-th position. Therefore constructing $D$ is reduced to solving $2 t$ triangular systems, each  of size $t \times t$, resulting in $O(t^3) = O(\hat{N}^\frac32)$ required steps.

\subsection{Step 2: Sampling Edges in the Separator} \label{subsec:step2}

Now, progressing iteratively, one finds $v \in P_3$ which is not yet paired and draw an edge emanating from it. Suppose that the edges, $e_1 = \{ v_{j_1}, v_{j_2} \}, ..., e_k = \{ v_{j_{2 k - 1}}, v_{j_{2 k}} \}$, are already sampled. We assume that by this point we have also computed LU-decomposition $A_k = [K^{-1}]_{j_1, ..., j_{2 k}}^{j_1, ..., j_{2 k}} = L_k U_k$ and we will update it to $A_{k + 1}$ when the new edge is drawn. Then
\begin{equation}
    \mathbb{P}(e_1, ..., e_k \in M) = \sqrt{| \det A_k |} \prod_{j = 1}^k c_{e_j}
\end{equation}

Next we choose $j_{2 k + 1}$ so that $v_{j_{2 k + 1}}$ is not saturated yet. We iterate over $v_{j_{2 k + 1}}$'s neighbors considered as candidates for becoming $v_{j_{2 k + 2}}$. Let $v_j$ to become the next candidate, denote $e_{k + 1} = \{ v_{j_{2 k + 1}}, v_j \}$. For $n \in \mathbb{N}$ let $\alpha(n) = n + 1$ if $n$ is odd and $\alpha(n) = n - 1$ if $n$ is even. Then the identity
\begin{equation}
    K^{-1} = [\overline{K}^{-1}]_{1, 2, ..., \hat{N}}^{\alpha(1), \alpha(2), ..., \alpha(\hat{N})},
    \label{eq:app_K^-1}
\end{equation}
follows from the definition of $\overline{K}$. One deduces from Eq.~(\ref{eq:app_K^-1})
\begin{equation*}
    A_{k + 1} = [K^{-1}]_{j_1, ..., j_{2 k + 1}, j}^{j_1, ..., j_{2 k + 1}, j} = [\overline{K}^{-1}]_{j_1, ..., j_{2 k + 1}, j}^{\alpha (j_1), ..., \alpha(j_{2 k + 1}), \alpha(j)}
\end{equation*}

Constructing $T^{**}$ one has $j_1, ..., j_{2 k + 1}, j, \alpha(j_1), ..., \alpha(j_{2 k + 1}), \alpha(j) > \hat{N} - t$. It means that $A_{k + 1}$ is a submatrix of $D$ with permuted rows and columns, hence $A_{k + 1}$ is known.

We further observe that
\begin{equation*}
    A_{k + 1} = \begin{bmatrix} A_k & y \\ r & d \end{bmatrix} = \begin{bmatrix} L_k & 0 \\ R & 1 \end{bmatrix} \begin{bmatrix} U_k & Y \\ 0 & z \end{bmatrix} = L_{k + 1} U_{k + 1}.
\end{equation*}
Therefore to update $L_{k + 1}$ and $U_{k + 1}$, one just solves the triangular system of equations $R U_k = r$ and $L_k Y = y$, where $R^\top, r^\top, Y, y$ are of size $2 k \times 2$ (this is done in $O(k^2)$ steps), and then compute $z = d - R Y$ which is of the size $2 \times 2$, then set, $u = \det z$.

The probability to pair $v_{j_{2 k + 1}}$ and $v_j$ is
\begin{align*}
    \mathbb{P}(e_{k + 1} \in M \, | \, e_1, ..., e_k \in M) &= \frac{\mathbb{P}(e_1, ..., e_{k + 1} \in M)}{\mathbb{P}(e_1, ..., e_k \in M)} \\
    &= \frac{ \sqrt{| \det A_{k + 1} |} \prod_{j = 1}^{k + 1} c_{e_j}}{ \sqrt{| \det A_k |} \prod_{j = 1}^k c_{e_j}} \\
    &= \frac{ c_{e_{k + 1}} \sqrt{ | u | | \det A_k |} }{\sqrt{ | \det A_k | }} \\
    &= c_{e_{k + 1}} \sqrt{| u |}
\end{align*}

Therefore maintaining $U_{k + 1}$ allows us to compute the required probability and draw a new edge from $v_{j_{2 k + 1}}$. By construction of $\hat{G}$, $v_{j_{2 k + 1}}$ has only $3$ neighbors, therefore the complexity of this step is $O(\sum_{k = 1}^{| P_3 |} k^2) = O(\hat{N}^\frac32)$ because $| P_3 | \leq 2^\frac32 \sqrt{\hat{N}}$.

\subsection{Step 3: Recursion} \label{subsec:step3}

Let $M_{sep} = \{ e_1, e_2, ... \}$ be a set of edges drawn on the previous step, and $\hat{V}_{sep}$ be a set of vertices saturated by $M_{sep}$, $P_3 \subseteq \hat{V}_{sep}$. Given $M_{sep}$, the task of sampling $M \in \text{PM}(\hat{G})$ such that $M_{sep} \subseteq M$ is reduced to sampling perfect matchings $M_1$ and $M_2$ over $\hat{G}(P_1 \setminus V_{sep})$ and $\hat{G}(P_2 \setminus V_{sep})$, respectively. Then $M = M_1 \cup M_2 \cup M_{sep}$ becomes the result of the perfect matching drawn from (\ref{eq:pmprobs}).

Even though only the first step of the Wilson's recursion was discussed so far, any further step in the recursion is done in exactly the same way with the only exception that vertex degrees may become less than $3$, while in $\hat{G}$ they are exactly $3$. Obviously, this does not change the iterative procedure and it also does not affect the complexity analysis.




\section{Random Graph Generation}

As our derivations cover the most general case of planar and $K_{33}$-free graphs, we want to test them on graphs which are as general as possible. Based on Lemma \ref{th:k33free} (notice, that it provides necessary and sufficient conditions for a graph to be $K_{33}$-free) we implement a randomized construction of $K_{33}$-free graphs, which is assumed to cover most general $K_{33}$-free topologies.

Namely, one generates a set of $K_5$'s and random planar graphs, attaching them by edges to a tree-like structure. For simplicity, we slightly relax the condition that random planar components should be triconnected (because it is not clear how to generate such graphs efficiently) and simply require the components to be biconnected. This can be interpreted as constructing $T$, where some neighbor planar nodes are merged (merging planar graphs results in another planar graph). We refer to such non-unique decomposition $T'$ as \textit{partially merged}. Inference and sampling algorithm suggested in Section \ref{sec:dynprog} is applied with no changes to the partially merged decomposition. Our generation process consists of the following two steps.
\begin{enumerate}
\item \textbf{Planar graph generation.} This step accepts $N \geq 3$ as an input and generates a normal biconnected planar graph of size $N$ along with its embedding on a plane. The details of the construction are as follows.

First, a random embedded tree is drawn iteratively. We start with a single vertex, on each iteration choose a random vertex of an already ``grown'' tree, and add a new vertex connected only to the chosen vertex. Items I-V in Fig.~\ref{fig:planargen} illustrate this step.

Then we triangulate this tree by adding edges until the graph becomes biconnected and all faces are triangles, as in the Subsection \ref{subsec:edg} (VI in Figure \ref{fig:planargen}). Next, to get a normal graph, we remove multiple edges possibly produced by triangulation (VII in Fig.~\ref{fig:planargen}). At this point the generation process is complete.

\item \textbf{$K_{33}$-free graph generation.} Here we take $N \geq 5$ as the input and generate a normal biconnected $K_{33}$-free graph $G$ in a form of its partially merged decomposition $T'$. Namely, we generate a tree $T'$ of graphs where each node is either a normal biconnected planar graph or $K_5$, and every two adjacent graphs share a virtual edge. 

The construction is greedy and is essentially a tree generation process from Step 1. We start with $K_5$ root and then iteratively create and attach new nodes. Let $N' < N$ be a size of the already generated graph, $N' = 5$ at first. Notice, that when a node of size $n$ is generated, it contributes $n - 2$ new vertices to $G$.

An elementary step of iteration here is as follows. If $N - N' \geq 3$, a coin is flipped and the type of new node is chosen - $K_5$ or planar. If $N - N' < 3$, $K_5$ cannot be added, so a planar type is chosen. If a planar node is added, its size is drawn uniformly in the range between $3$ and $N - N' + 2$ and then the graph itself is drawn as described in Step 1. Then we attach a new node to a randomly chosen free edge of a randomly chosen node of $T'$. We repeat this process until $G$ is of the desired size $N$. Fig.~\ref{fig:k33freegen} illustrates the algorithm.
\end{enumerate}

To obtain an Ising model from $G$, we sample pairwise interactions for each edge of $G$ independently from $\mathcal{N}(0, 0.1^2)$.

Notice that the tractable Ising model generation procedure is designed in this section solely for the convenience of testing and it is not claimed to be sampling models of any particular practical interest (e.g. in statistical physics or computer science). 

\begin{figure*}[!h] \centering
\centering
\includegraphics[width=0.95\linewidth]{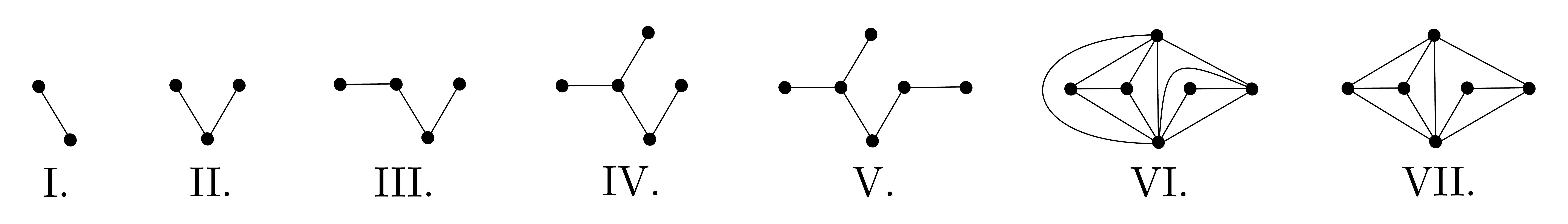}
\caption{Steps of planar graph generation. I-V refers to random tree construction on a plane, VI is a triangulation of a tree, VII is a result after multiple edges removal.}
\label{fig:planargen}
\end{figure*}

\begin{figure*}[!h] \centering
\centering
\includegraphics[width=0.95\linewidth]{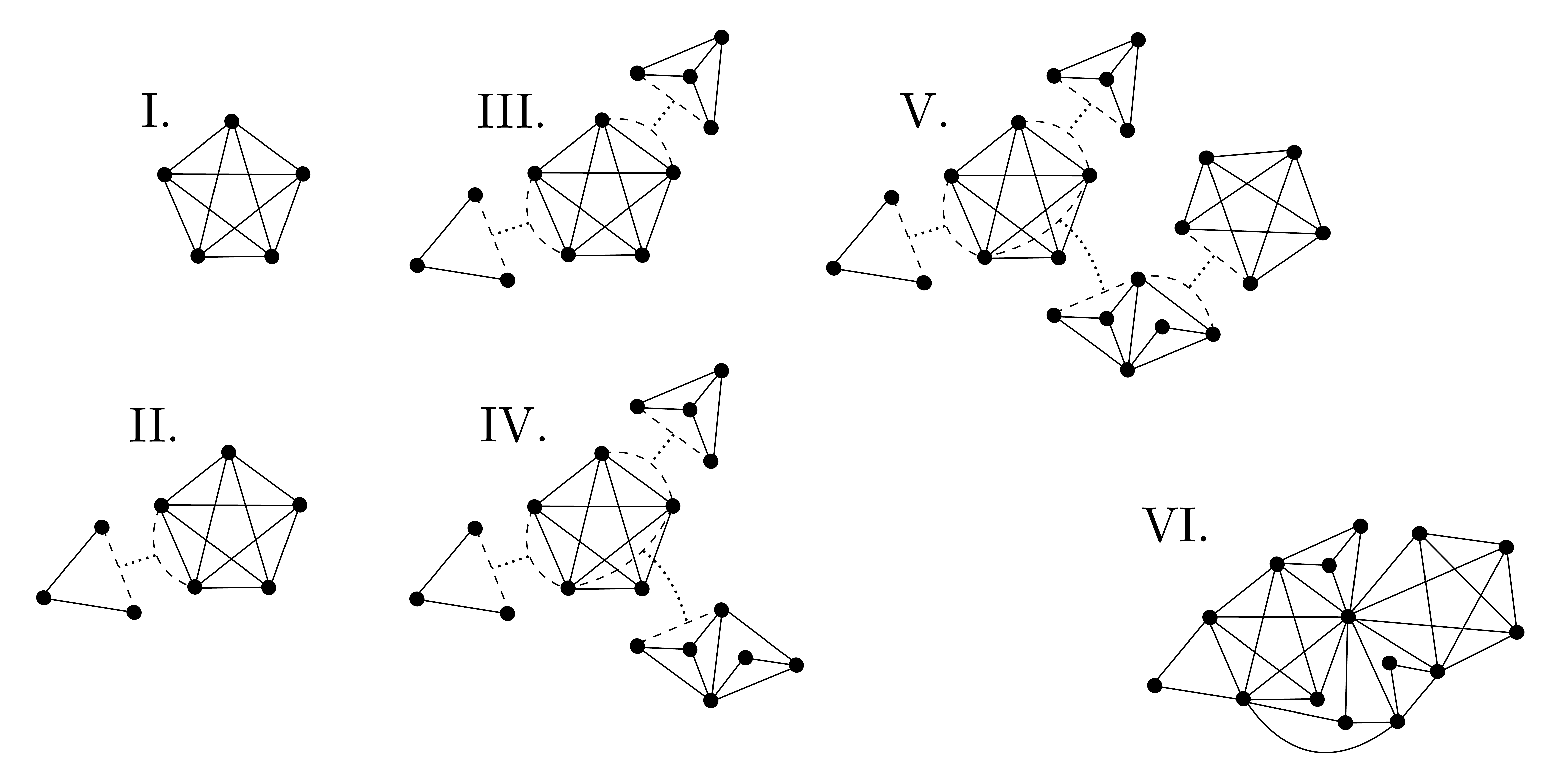}
\caption{Generation of $K_{33}$-free graph $G$ and its partially merged decomposition $T'$. Starting with $K_5$ (I), new components are generated and attached to random free edges (II-V). VI is a result graph $G$ obtained by merging all components in $T'$.}
\label{fig:k33freegen}
\end{figure*}

\section{Future Work}

We conclude by discussing some future research directions:
\begin{itemize}
    \item The class of models considered in the manuscript can be extended even further towards $K_{33}$-free generalizations of (a) the so-called outerplanar graphs, which can then be used for approximate inference and efficient learning in the spirit of \cite{07GJ} and \cite{johnson} respectively; and (b) graphs embedded in the surfaces of $O(1)$ genus \cite{regge,gallucio,cimasoni1,cimasoni2}.
    \item This manuscript was motivated by a larger task of using efficient inference and learning over the most general $K_{33}$-graphs for constructing more general (and thus, hopefully, more powerful) alternatives to traditional Neural Networks for efficient learning. 
\end{itemize}

\end{appendices}

\end{document}